\documentclass[11pt,a4paper]{article}
\usepackage{amsmath,amssymb,amsfonts,amsthm}
\usepackage[pdftex]{graphicx, color}

\newcommand{\E}{{\bf{E}}}
\newcommand{\PP}{{\bf{P}}}

\newcommand{\Bin}{{\bf{Bin}}}

\newtheorem{tm}{Theorem}

\newtheorem{lem}{Lemma}

\begin{document}

\bibliographystyle{plain}
\parindent=0pt
\centerline{\LARGE \bfseries Clustering function: a measure of social influence}

\par\vskip 3.5em

\centerline{
Mindaugas  Bloznelis\footnote{Faculty of Mathematics and Informatics, Vilnius University, 
03225 Vilnius, Lithuania}\footnote{Corresponding author mindaugas.bloznelis@mif.vu.lt}
and 
Valentas Kurauskas$^1$
}


\bigskip





\begin{abstract}
A commonly used characteristic of  statistical 
dependence of adjacency relations  in  real networks, the 
  clustering coefficient, evaluates  chances that two neighbours of a 
given vertex are adjacent. An extension is obtained by 
 considering conditional probabilities that two randomly chosen vertices 
are adjacent given that they have $r$ common neighbours. We denote
 such probabilities  $cl(r)$ and call  $r\to cl(r)$ the clustering function. 
 We compare clustering functions of several networks having non-negligible clustering 
coefficient. They show similar patterns and surprising regularity. 
We establish a first order asymptotic (as the number of vertices $n\to+\infty$) of 
the clustering function 
of  related random intersection graph models admitting nonvanishing clustering 
coefficient and
asymptotic degree distribution having a
finite second moment. 
\par
\end{abstract}

\smallskip
{\bfseries key words:}  clustering coefficient, social network, intersection graph, power law
\par\vskip 2.5em

2000 Mathematics Subject Classifications: 91D30,   05C80,  05C07, 91C20

\par\vskip 2.5em

\section{Introduction}
Our study is motivated by the following
question: given two vertices of a network, the presence of 
how many common neighbours would imply with certainty  
that these two vertices are adjacent. 
A "softer" question is about the probability that two vertices 
with (at least) $r$ common neighbours establish a link. The answer 
is given by the clustering functions (\ref{Cl1}) and (\ref{Cl2}).

Let ${\cal G}=({\cal V}, {\cal E})$ be a finite graph on vertex set
 ${\cal V}$ and with edge set ${\cal E}$. 
The number of neighbours of a vertex $v$ is denoted $d(v)$.
The number of common neighbours of vertices $v_i$ and $v_j$ is denoted $d(v_i,v_j)$.
We are interested 
in the fraction of adjacent pairs $v_i\sim v_j$ among all pairs 
$\{v_i,v_j\}\subset {\cal V}$ having (at least) $r$ common neighbours.
Here and below '$\sim$' denotes the adjacency relation of ${\cal G}$.
More formally, let us consider the random pair of distinct vertices $\{v_1^*, v_2^*\}$ 
drawn from ${\cal V}$
uniformly at random. Define the clustering functions of ${\cal G}$ 
\begin{eqnarray}\label{Cl1}
&&
r\to cl_{\cal G}(r):=\PP(v_1^*\sim v_2^*|\,d(v_1^*,v_2^*)=r),
\\
\label{Cl2}
&&
r\to Cl_{\cal G}(r):=\PP(v_1^*\sim v_2^*|\,d(v_1^*,v_2^*)\ge r).
\end{eqnarray}
In the case of a social network (\ref{Cl1}), (\ref{Cl2})  could be interpreted as measures of social influence or 
pressure exercised by the
neighbours on a  pair of actors to   establish a communication link. 
We remark that characteristics (\ref{Cl1}) and (\ref{Cl2}) are related  to the clustering coefficient of ${\cal G}$. We recall its definition for convenience.
 Let $(v_1^*, v_2^*, v_3^*)$ be an ordered  triple of distinct vertices drawn from ${\cal V}$
uniformly at random. The conditional probability that $v_1^*$ is adjacent to $v_2^*$, given that
$v_1^*$ and $v_2^*$ are both adjacent to $v_3^*$, is called the (global) clustering coefficient (\cite{Barrat2000}, \cite{Newman2001}, \cite{Newman2003}, \cite{storgatz1998}).
We denote it $C=C_{\cal G}=\PP(v_1^*\sim v_2^*|v_1^*\sim v_3^*, v_2^*\sim v_3^*)$.

In this paper we study clustering functions  first by  considering  empirical data and then 
by a rigorous analysis of 
related random graph models.


We consider clustering function (\ref{Cl1}) of  real networks admitting positive 
clustering
coefficient: the actor network, where two actors are declared adjacent whenever 
they have acted
in the same film (\cite{actornetwork}), and the Facebook network 
(\cite{Bakshy}, \cite{Foudalis2011}, \cite{Traud}).
 We remark that empirical plots  show similar pattern and surprising regularity.

 Our choice of the random  graph model
is motivated by an observation of Newman et al. \cite{Newman+W+S2002} that 
the clustering property of some social networks (so called affiliation networks) could be explained by the presence of a bipartite graph structure. For example, the bipartite graph, where actors are linked to films, defines the actor network.  
It seems reasonable that 
a bipartite graph structure might also be helpful in   
explaining (at least to some extent)
the adjacency relations of  Facebook network: two 
members become adjacent because they share some common interests/attributes. 

We secondly consider clustering function (\ref{Cl1}) of a  
random intersection graph, where vertices (actors) are prescribed attribute 
sets independently at random and two vertices are declared  
adjacent whenever they share at least one common attribute
(\cite{karonski1999}, \cite {godehardt2003},  see also 
\cite{Barbour2011}, \cite{Guillaume+L2004}). Random intersection graphs  
are relatively simple objects and for them rigorous mathematical results can be obtained. 
We evaluate
the probabilities $\PP(v_1^*\sim v_2^*|\,d(v_1^*,v_2^*)=r)$, $r=0,1,2,\dots$ 
for a random intersection graph in Sect. 3 below.

\section{Clustering functions: empirical results}

In Figure 1 we plot clustering functions (\ref {Cl1}) and (\ref{Cl2}) of 
three drama actor networks: 
the English actor network with $n=402622$
actors, $m=66127$ films  and the clustering coefficient $C=0.32$ 
(the clustering coefficients here and below are rounded up to 2 decimal places), 
the French actor network 
with $n=43204$
actors, $m=5629$ films and the clustering coefficient $C=0.30$ 
and  
the Russian actor network with $n = 9880$ actors, $m=2459$ films and $C=0.44$. 
The data has been obtained from \cite{actornetwork}. 
\begin{figure}
    \begin{center}
        \includegraphics[scale=0.5]{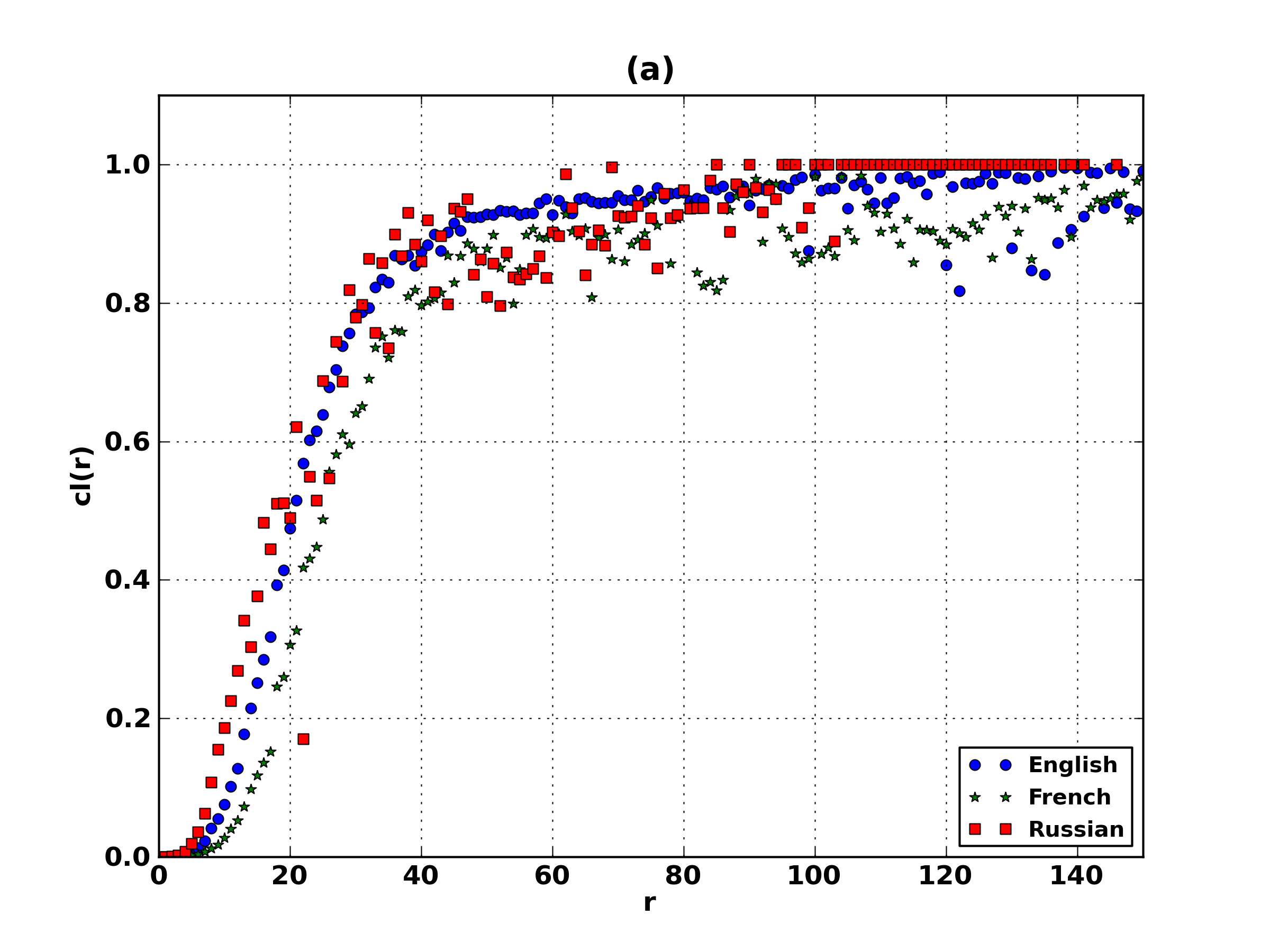}
        \includegraphics[scale=0.5]{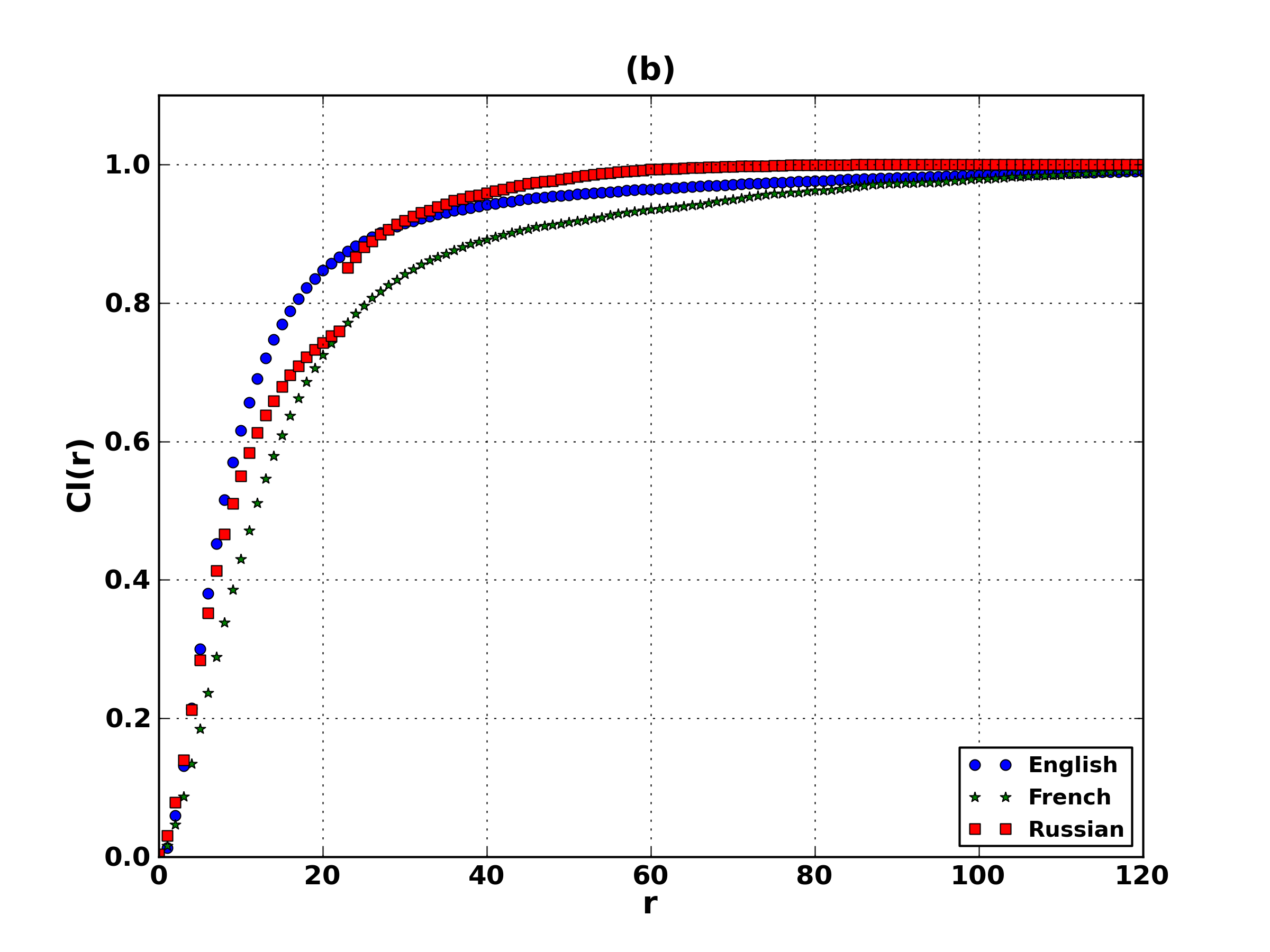}
    \end{center}
    \caption{\small     Clustering functions for three actor networks: (a) $cl(r)$, (b) $Cl(r)$.}
 \label{fig.1}
\end{figure}
In Figure 2 we plot clustering function (\ref{Cl1})  of three networks 
describing relations between community members at three  
different universities (data from \cite{Traud}): 
the first network has $n=17425$ vertices and the clustering 
coefficient $C=0.16$ ($\bullet$ blue graph);
the second network has $n=9414$ vertices and 
the clustering coefficient $C=0.15$ ($\star$ green graph);
the third network has $n=6596$ vertices and the clustering 
coefficient $C=0.16$ ({\tiny{$\blacksquare$}} red graph).
\begin{figure}
    \begin{center}
        \includegraphics[scale=0.5]{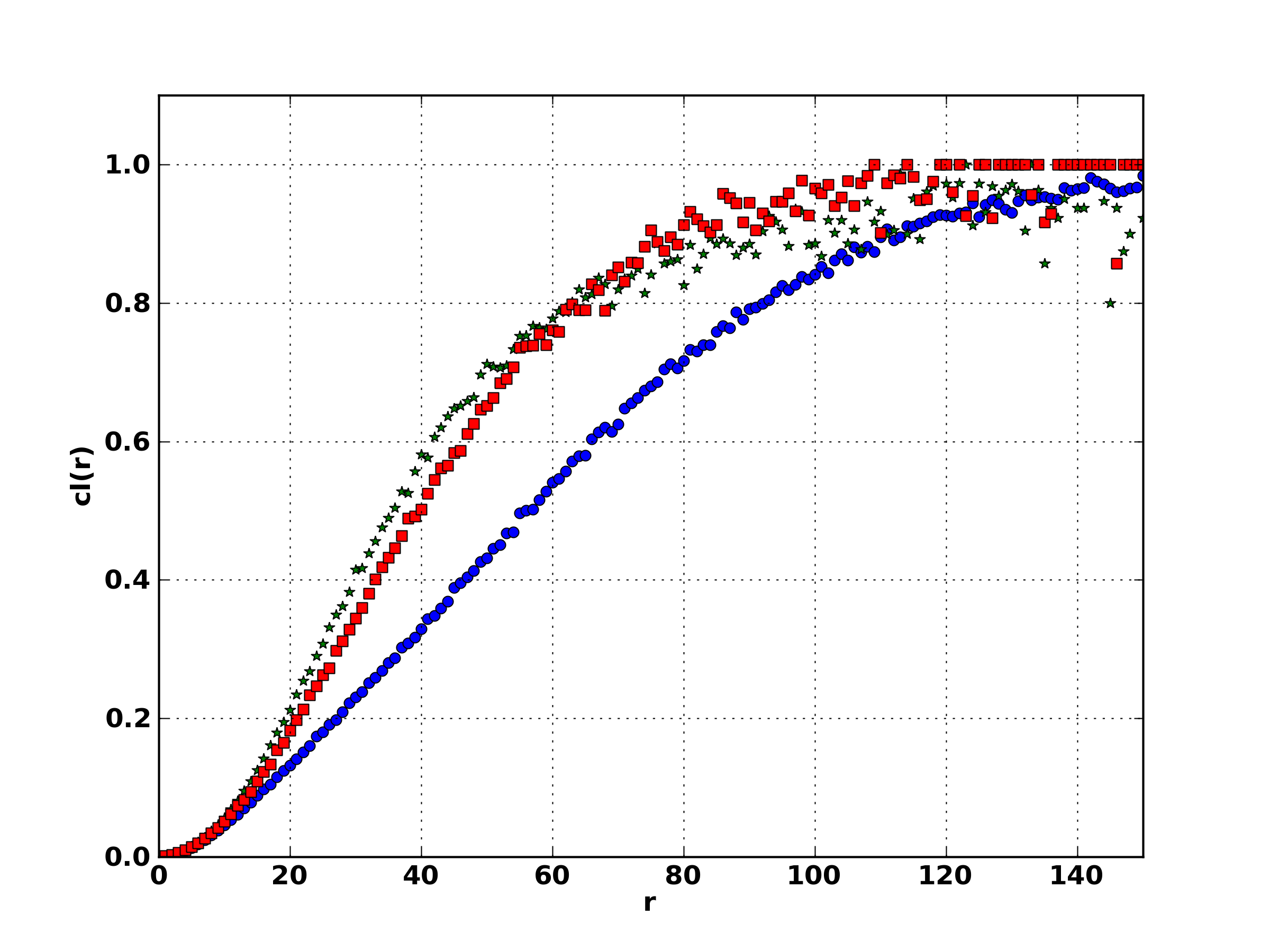}
    \end{center}
    \caption{\small     Clustering functions of three  university networks.}
 \label{fig.2}
\end{figure}

\section{Clustering functions of random intersection graphs}

Vertices $v_1,\dots, v_n$ of an  intersection graph are represented 
by subsets $D_1,\dots, D_n$ of a given ground set 
$W=\{w_1,\dots, w_m\}$. Elements of $W$ are called attributes or keys. 
Vertices $v_i$ and $v_j$ are declared adjacent if 
 $D_i\cap D_j\not=\emptyset$. 
 The adjacency relations of such an intersection graph resemble those of some
real networks,  e.g., the collaboration network, where  authors are declared 
adjacent whenever they have co-authored a paper, or  the 
actor network, where two actors are linked by an edge whenever they have 
acted in the same film. 
 Random intersection graphs 
 have 
 attracted considerable attention in the recent literature, see, e.g., 
\cite{behrisch2007},
\cite{Blackburn2009},
\cite{Bloznelis2011+}
\cite{Britton2008},
\cite{eschenauer2002},
\cite{Rybarczyk2011},
\cite{Spirakis2011},
\cite{Yagan2009}.
They admit a power law degree distribution and tunable clustering. 
We consider two models of random intesection graphs: the active graph and the inhomogeneous graph.

{\bf Active graph.} In the {\it active} random intersection graph  $G_1(n,m,P)$
every vertex $v_i\in V=\{v_1,\dots, v_n\}$ selects its attribute set $D_i$ 
independently at random (\cite{godehardt2003}, \cite{karonski1999}).
We assume for simplicity that  independent random sets
$D_1,\dots, D_n$ have the same probability distribution  
\begin{equation}\label{DP}
 \PP(D_i=A)={\tbinom{m}{|A|}}^{-1}P(|A|),
\qquad {\text{ for \ \  any}}
\quad
A\subset W.
\end{equation}
 In particular, all attributes have equal 
probabilities to be selected.
Here $P$ is the common probability distribution of  
the sizes of selected sets $X_i:=|D_i|$ (for each $i=1,\dots, n$ we have
$\PP(X_i=k)=P(k)$, $k=0,1,\dots m$).  
We remark that $X_1,\dots, X_n$ are independent random variables taking values in $\{0,1,\dots, m\}$.

We study the clustering function 
\begin{eqnarray}\label{Cl1+}
r\to cl(r)=\, \PP(v_1^*\sim v_2^*|\,d(v_1^*,v_2^*)=r)\,=\PP(v_1\sim v_2|\,d(v_1,v_2)=r)
\end{eqnarray}
of a sparse random intersection graph with  large number of vertices. We remark, that 
the second identity of (\ref{Cl1+}) follows from the fact that the probability  
distribution of $G_1(n,m,P)$ is invariant under permutation of its vertices.
By sparse we mean
that the number of edges scales as the number of vertices $n$ as  $n\to+\infty$.
It is convenient to consider a sequence of random intersection graphs $\{G_{(n)}\}_n$, where
 $G_{(n)}=G_1(n,m, P)$ and where  $m=m_n$ and $P=P_n$ both depend on $n$. 
We remark that $\{G_{(n)}\}_n$ is  a sequence of 
sparse random graphs whenever
the size $X_1$ of the typical random set is of order $(m/n)^{1/2}$  as $m,n\to\infty$ 
(\cite{Bloznelis2008}).  
Furthermore, assuming
 that

(i)   $X_1\sqrt{n/m}$ converges in distribution to some random variable $Z$;

(ii)  $\E Z<\infty$ and $\E X_1\sqrt{n/m}$ converges to $\E Z$

one obtains the  asymptotic degree distribution  of $\{G_{(n)}\}$ 
\begin{equation}\label{adegree}
 \lim_{n\to+\infty}\PP(d(v_1)=k)=\E e^{-Z\,\E Z}(Z\,\E Z)^k/k!,
\qquad 
{\text{for}}
\qquad
k=0,1,\dots,
\end{equation}
see \cite{Bloznelis2008}, \cite{Bloznelis2011+}, \cite{Deijfen}, \cite{stark2004}.  
Here $d(v)$ denotes the degree of a vertex $v$. We remark that a heavy tailed 
distribution of 
$Z$ yields a heavy tailed asymptotic degree distribution (\ref{adegree}).
Along with the first moment condition (ii) we shall  also consider the $r-$th  moment condition

(ii-r) $\E Z^r<\infty$ and $\E (X_1\sqrt{n/m})^r$ converges to $\E Z^r$.

 We denote 
$z_r=\E Z^r$ and  
 $\delta_r=\E d_*^r$  where $d_*$ is a random variable with
the asymptotic degree distribution
$\PP(d_*=k)=\E e^{-z_1Z}(z_1Z)^k/k!$, $k=0,1,\dots$.
We assume below that $\E Z>0$, i.e., that the asymptotic degree 
distribution is non-degenerate. Furthermore, we assume for convenience that the ratio 
$\beta_n=m/n$ tends 
to some $\beta\in (0,+\infty]$
as $n\to+\infty$.

An important property of the active random intersection graph is that the 
adjacency relations  are statistically dependent events. 
In particular, the clustering coefficient 
\begin{displaymath}
\alpha=\alpha(G_{(n)})=\PP(v_1\sim v_2|v_1\sim v_3, v_2\sim v_3) 
\end{displaymath}
of a sparse random intersection graph $G_{(n)}$
is bounded away from zero
as $n \to+\infty$ provided that the second moment of the degree distribution is finite and  
$\beta<\infty$ (\cite{Bloznelis2011+}, \cite{Deijfen}). In this case we have 
(see (\cite{Bloznelis2011+}, \cite{Deijfen}) 
\begin{eqnarray}\label{cc}
\alpha
=
\beta^{-1/2}\delta_1^{3/2}(\delta_2-\delta_1)^{-1}+o(1)
\qquad{\text{as}}
\qquad
n\to+\infty.
\end{eqnarray}
We remark that for $\beta=+\infty$ we have  $\alpha=o(1)$. For comparison,  the  
(unconditional) edge probability 
$p_e:=\PP(v_1\sim v_2)$ satisfies for any 
$\beta\in (0,+\infty]$, see, e.g., \cite{Bloznelis2011+}, 
\begin{displaymath}
 p_e=\delta_1n^{-1}+o(n^{-1})=O(n^{-1}).
\end{displaymath}

 
 Theorems \ref{T1} and \ref{T2} show a first order asymptotics of the 
conditional probabilities
$cl(r)$  
as $n\to+\infty$
 in the cases where $\beta<\infty$ and  $\beta=\infty$, respectively.



\begin{tm}\label{T1} Let $m,n\to\infty$. 
Assume that (i), (ii-2)  hold and $\E Z>0$. Suppose that  $\beta_n\to \beta \in(0,+\infty)$.
Denote $\Lambda=\sqrt{\delta_1/\beta}$. We have
\begin{equation}
cl(r)=
\begin{cases}
p_ee^{-\Lambda}(1+o(1)),
\qquad \ \ \
\quad
r=0;
\\
\label{b2-4}
\frac{\alpha}{\alpha+(1-\alpha)e^{\Lambda}}(1+o(1)),
\quad \quad
\
r=1;
\\
1-o(1),
\qquad
\qquad
\quad
\quad
\quad
\
r\ge 2.
\end{cases}
\end{equation}
\end{tm}

Empirical results  of simulated random intersection graphs show that
the convergence to the ``limiting shape``  in (\ref{b2-4}) is rather slow,
 see  Figures 3 and 4 below.

\begin{tm}\label{T2} Let $m,n\to\infty$. 
Assume that (i), (ii-2)  hold and $\E Z>0$. Suppose that $\beta_n\to+\infty$.
 We have  
\begin{equation}\label{b2-1}
cl(r)=
\begin{cases}
p_e(1+o(1)),
\qquad \
\qquad
\qquad
\qquad
\qquad
\ \
r=0;
\\
\beta_n^{-1/2}(z_1/z_2)(1+o(1))+O(n^{-1}),
\quad \quad
\
r=1.
\end{cases}
\end{equation}
In particular, $cl(0)=O(n^{-1})$ and $cl(1)=o(1)$. Furthermore, we have
\begin{equation}\label{b3-2}
cl(2)=
\begin{cases}
1+o(1),
\qquad 
\qquad
\qquad
\qquad
{\text{for}}
\qquad
\beta_n/n\to 0;
\\
\frac{1}{1+\beta_*(\delta_2-\delta_1)^4\delta_1^{-6}}+o(1),
\qquad
\ 
{\text{for}}
\qquad
\beta_n/n\to\beta_*\in (0,+\infty);
\\
o(1),
\qquad
\qquad
\qquad\qquad
\quad
\ \
{\text{for}}
\qquad 
\beta_n/n\to+\infty.
\end{cases}
\end{equation}
Assuming, in addition, that 
$\beta_n^k=o(n)$ for each $k=1,2,3\dots$,
we obtain
\begin{equation}\label{b2-3}
cl(r)=1+o(1),
\qquad
{\text{for}}
\qquad
r=2,3,\dots. 
\end{equation}
\end{tm}
We conclude from (\ref{b2-4}), (\ref{b2-1})  that 
 edge dependence measures $cl(1)$ and $\alpha$ are closely related.
In particular, we have
$cl(1)=1-o(1)\Leftrightarrow \alpha=1-o(1)$ 
and $cl(1)=o(1)\Leftrightarrow \alpha=o(1)$.
Furthermore, (\ref{b3-2}) tells us that the characteristic $cl(2)$ is  able to
distinguish between the cases $\beta_n=o(n)$ and $n=o(\beta_n)$.
Finally, (\ref{b2-3}) tells us that any $cl(r)$, $r=1,2,\dots$ can't distinguish between  
 sequences
  $\{\beta_n\}$ and $\{\beta_n'\}$ growing slower that any power of $n$ 
(take $\beta_n=\ln n$ and $\beta_n'=\ln^2n$, for example).

{\it Remark 1.}  It is likely that (\ref{b3-2}) can be extended to an arbitrary $r$ as follows
\begin{displaymath}
cl(r)=
\begin{cases}
1+o(1),
\qquad 
\qquad
\quad
\,
{\text{for}}
\qquad
\beta_n/n^{4-2r^{-1}}\to 0;
\\
c(r,\beta_*)+o(1),
\qquad
\  \,
{\text{for}}
\qquad
\beta_n/n^{4-2r^{-1}}\to\beta_*\in (0,+\infty);
\\
o(1),
\qquad
\quad\qquad
\quad
\ \ \,
{\text{for}}
\qquad 
\beta_n/n^{4-2r^{-1}}\to+\infty.
\end{cases}
\end{displaymath}
Here $c(r,\beta_*)=(\beta_*^{r/2}z_1^{-r-2}z_2^rz_r^2+1)^{-1}$. We note that numbers $z_i=\E Z^i$
can be expressed in terms of moments of  the asymptotic degree distribution (\ref{adegree}).

Proofs of Theorems \ref{T1} and \ref{T2} are given in Sect. 5.

Fig. 3  illustrates the convergence to a step function shown by Theorem \ref{T1}. Here 
we plot clustering function 
(\ref{Cl1}) of simulated 
random intersection graphs $G_i=G(n_i,m_i,P)$, where $n_i=m_i=10^{2+i}$, 
$i=1,2,3$, and $P(10)=1$.

In Fig. 4 we plot (\ref{Cl1}) for simulated
random intersection graphs $G_i=G(n,m,P_i)$, where $n=m=10^4$
and $P_i(3^i)=1$, $i=1,2,3$.
Fig. 4 illustrates the influence of the size of random sets. 

\begin{figure}
    \begin{center}
        \includegraphics[scale=0.5]{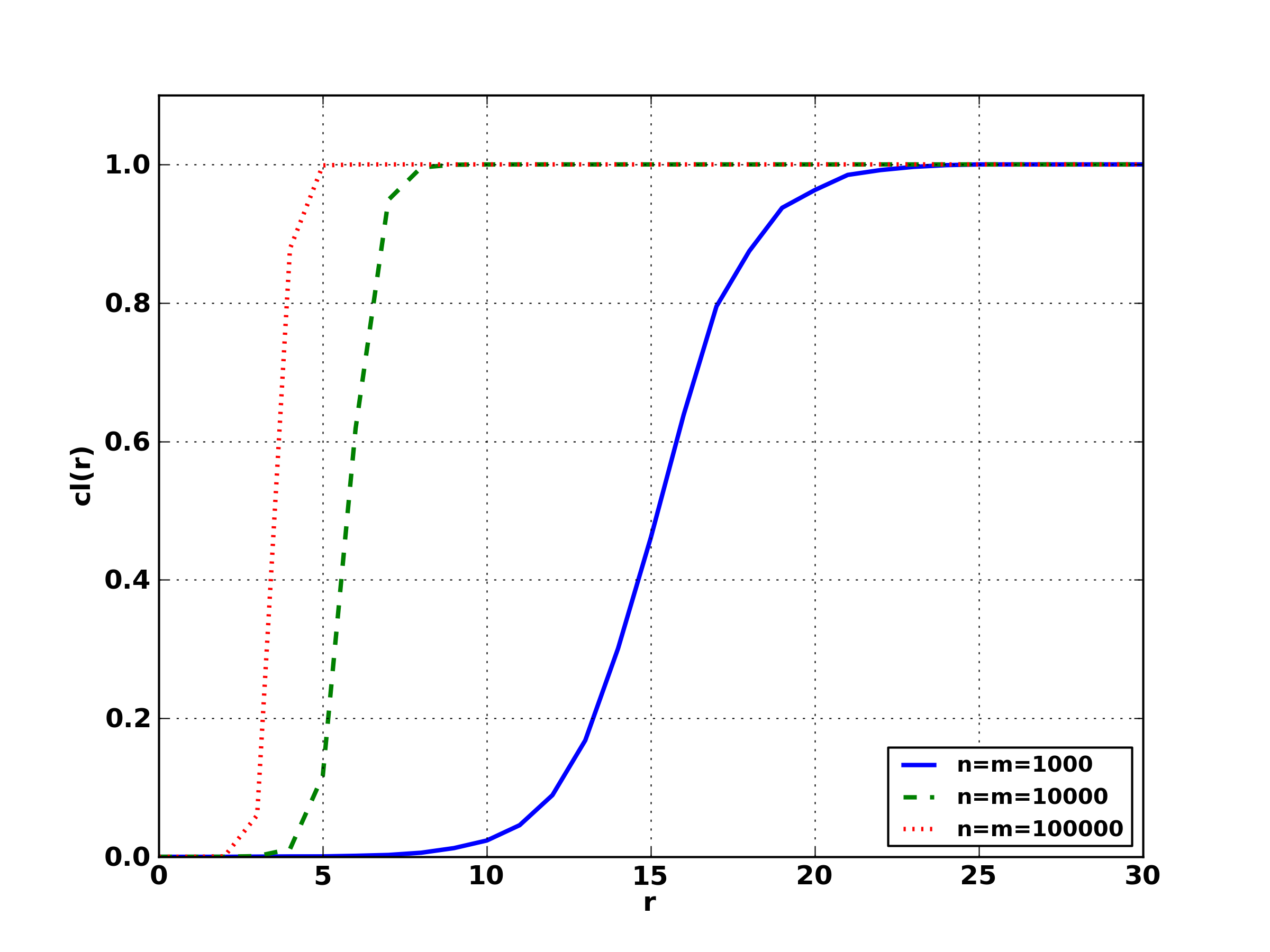}
    \end{center}
    \caption{\small Convergence to the step function for random intersection graphs with all sets of size 10.}
 \label{fig.3}
\end{figure}

\begin{figure}
    \begin{center}
        \includegraphics[scale=0.5]{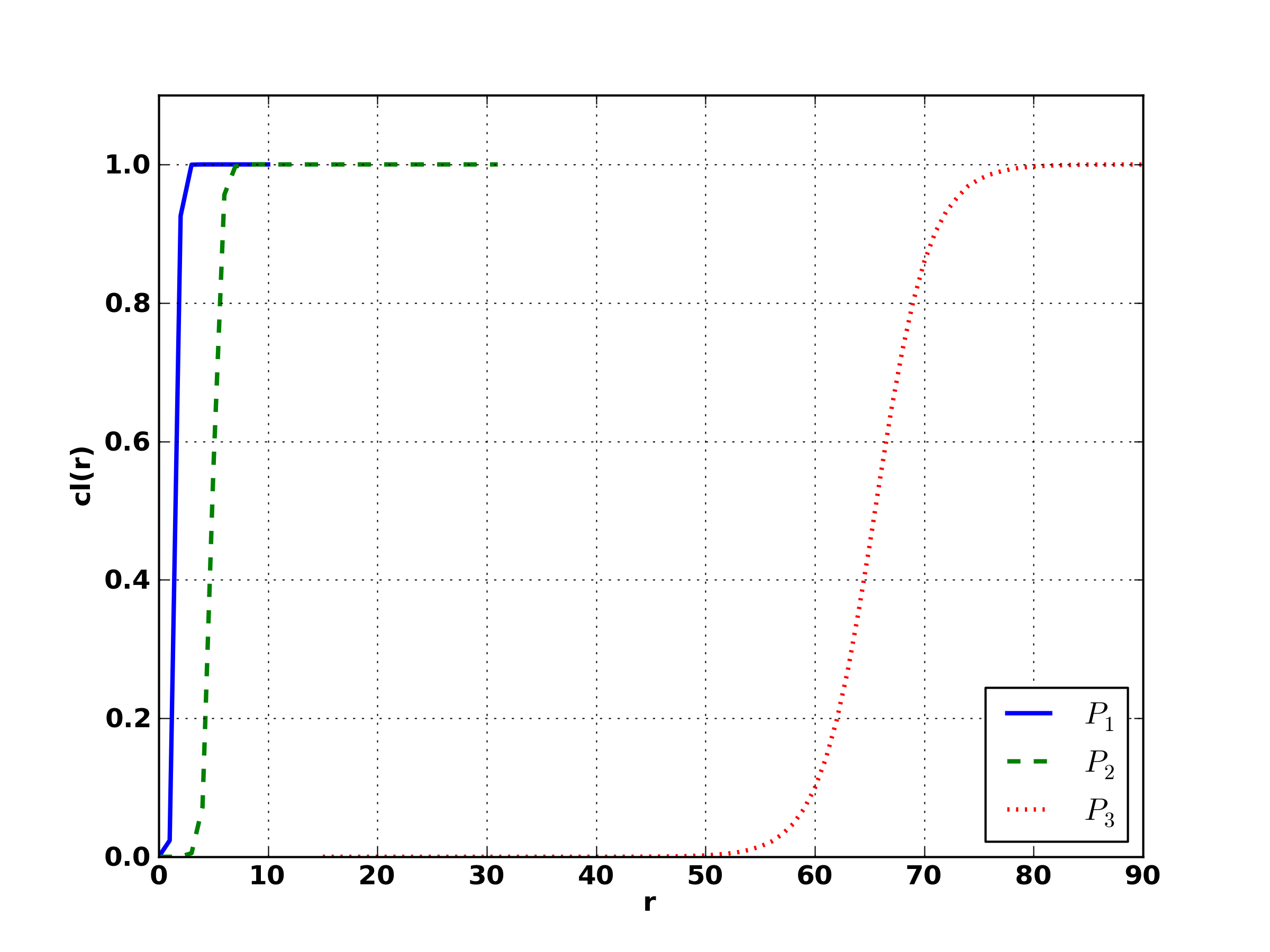}
    \end{center}
    \caption{\small Clustering function of random intersection graphs with $n = m = 10000$ and $P_i(3^i) = 1$, $i=1,2,3$.}
 \label{fig.4}
\end{figure}

{\bf Inhomogeneous graph}. The 
 {\it inhomogeneous} random intersection graph  $G_1(n,m,P_1,P_2)$ on the vertex 
set $V=\{v_1,\dots, v_n\}$ is obtained as follows.
We first generate independent random variables $A_1,\dots, A_n, B_1,\dots, B_m$ such that
each $A_i$ has the  probability distribution $P_1$ and each $B_j$ has the probability distribution
$P_2$. Then, conditionally on the realized values $\{A_i,B_j\}_{i,j=1}^{n,m}$, we include the attribute $w_j\in W$
in the set $D_i$ with probability $p_{ij}=\min\{1, A_iB_j(nm)^{-1/2}\}$ independently 
for each $i$ and $j$ (see \cite{Barbour2011}, 
\cite{BloznelisD2012+},  \cite{Bradonjic2010}, \cite{Shang2010}).  Our motivation of 
studying 
this random graph model  is that its clustering function approximates 
empirical 
data remarkably well, see Figure 8 below.
 
We consider a sequence of inhomogeneous intersection graphs 
$\{{\tilde G}_n=G_1(n,m,P_1,P_2)\}$,
where  $P_1$, $P_2$
remain fixed while
 $m=m_n$ and $n$ tend to infinity. We  
denote $a_k=\E A_1^k$ and $b_k=\E B_1^k$. 
A simple calculation  shows (see Section 5 below)
that the edge 
probability $p_e=\PP(v_1\sim v_2)$ of $G_1(n,m,P_1,P_2)$ satisfies
\begin{equation}\label{p-e-inh}
p_e=a_1^2b_2n^{-1}+o(n^{-1}). 
\end{equation}
Hence, $\{{\tilde G}_n\}$ is a sequence of sparse graphs. We remark that this sequence admits
a power law asymptotic degree distribution  \cite{BloznelisD2012+}.

In Theorem \ref{T3} below we show a first order asymptotics of the 
clustering function $cl(\cdot)$ in the case where the ratio $\beta_n=m/n$ 
has a non-zero finite 
limit. In addition, we show that ${\tilde G}_n$ admits
a nonvanishing clustering  coefficient
 $\alpha=\alpha({\tilde G}_n)=\PP(v_1\sim v_2|v_1\sim v_3, v_2\sim v_3)$.

\begin{tm}\label{T3} Let $m,n\to\infty$. 
Assume that $0<\E A_1^2<\infty$ and $0<\E B_1^3<\infty$. Suppose that 
$\beta_n\to \beta \in(0,+\infty)$.
Then we have
\begin{equation}\label{clcoef}
 \alpha=\frac{b_3\kappa}{ b_3\kappa+\sqrt{\beta}}+o(1)
 \qquad\qquad
 \qquad\qquad \ \
\end{equation}
 and 
\begin{equation}
cl(r)=
\begin{cases}
a_1^2b^*_2n^{-1}(1+o(1)),
\qquad 
\quad
\ 
\
r=0;
\\
\label{NH}
\frac{b^*_3\kappa}{b^*_3\kappa+\sqrt{\beta}}(1+o(1)),
\qquad 
\qquad
r=1;
\\
1-o(1),
\qquad
\qquad
\quad
\quad
\quad
\ \ \
r\ge 2.
\end{cases}
\end{equation}
Here $\kappa=a_1a_2^{-1}b_2^{-2}$ and \ $b^*_k=\E B_1^ke^{-a_1B_1/\sqrt{\beta}}$. 
\end{tm}

The proof of Theorem \ref{T3} is given in Sect. 5. 


\section{Discussion}

The first order asymptotics (\ref{b2-4}), (\ref{b3-2}), (\ref{b2-3}) and  
(\ref{NH}) suggests that 
the clustering function $cl(\cdot)$ of a  large intersection graph with a square integrable 
asymptotic degree distribution 
can be
approximated by a step - like function. 
Furthermore, $cl(1)$ is closely related to the clustering 
coefficient. 

Simulations in Figures 3 and 4 show that the convergence in (\ref{b2-4}) can 
be rather slow and we observe a sigmoid function approximation of the step function. 
Furthermore, the larger is the average 
degree, the more remote is the ``step`` from the origin 
and the more gradual is the slope of the clustering function. 

Clustering functions of real networks considered in Figures 1 and 3 have even more gradual 
slope, a phenomena perhaps related to the inhomogeneity 
of the degree sequence. We remark that the  actor network and the Facebook are considered
as having power law degree sequences 
which do not admit a 
finite (theoretical) second moment, see, 
e.g., \cite{Durret2007}, \cite{Foudalis2011}.
In order to
learn more about  the influence of the 
inhomogeneity of the degree sequence on the slope of the clustering function $r\to cl(r)$ 
we select various subnetworks of  real networks according to certain regularity conditions
satisfied by their degree sequences. We observe that the  inhomogeneity (heavy tail) 
of the degree sequence 
affects the slope of the clustering function: the heavier the tail 
the more gradual is the slope of the clustering function. We illustrate these 
observations in Figures 5 and 6.

Figure 5 plots clustering function (\ref{Cl1}) of subgraphs of the first university 
network (see Sect 2.) sampled as follows.
${\cal G}_1$ is the subgraph
that includes all vertices  of degree not larger than $50$. It has $n_0=7165$ vertices. 
${\cal G}_2$ is a subgraph  induced by $n_0$ vertices
drawn uniformly at random (without replacement) from the vertices of degree not 
larger than $150$. ${\cal G}_3$ is a subgraph of induced by $n_0$ vertices
drawn uniformly at random (without replacement) from the set of all vertices. 
Now all three graphs have the same number of vertices.
\begin{figure}
    \begin{center}
        \includegraphics[scale=0.5]{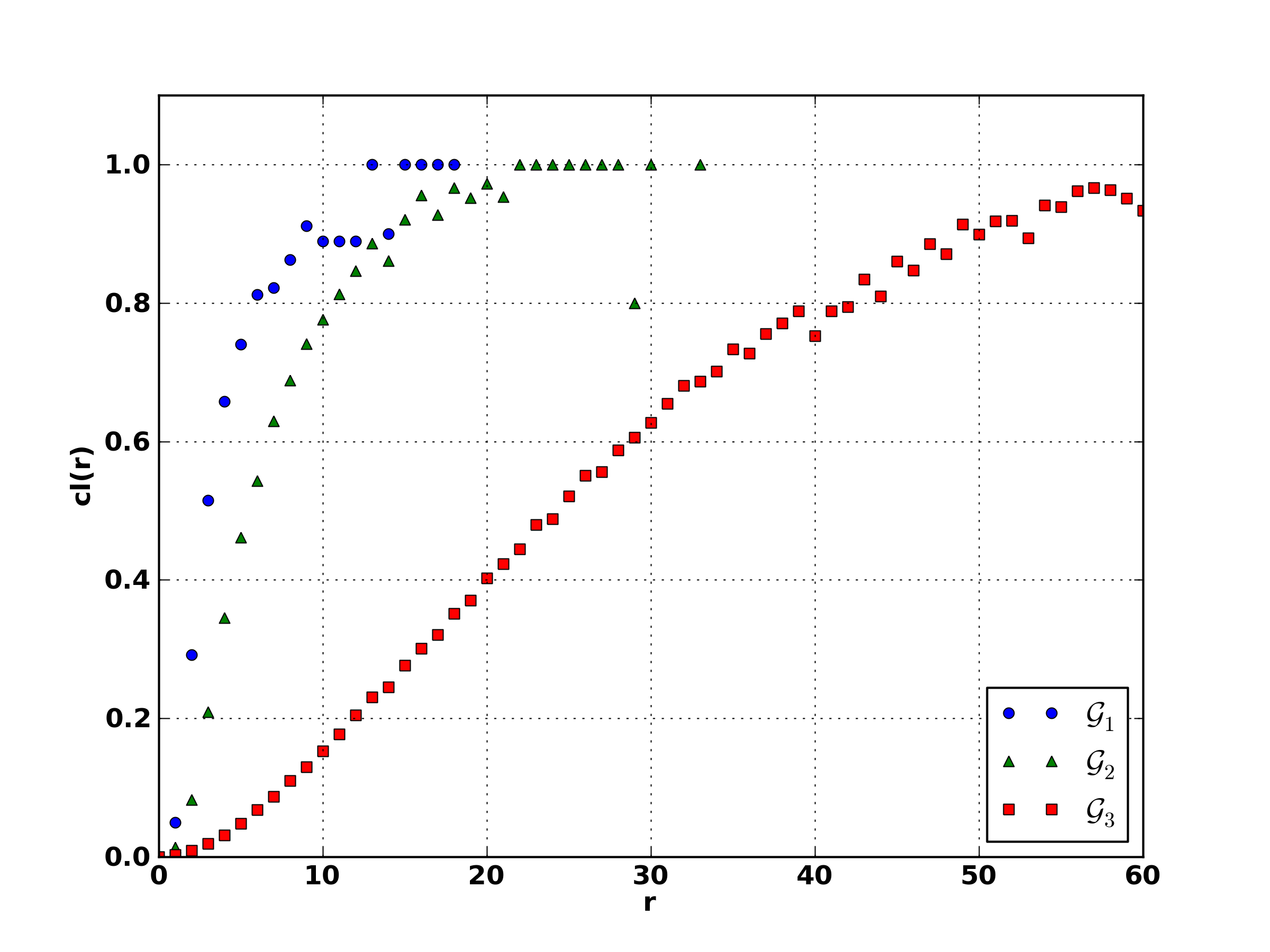}
    \end{center}
    \caption{\small Sampling subgraphs with degree constraints.}
 \label{fig.5}
\end{figure}

In Figure 6 we plot two subgraphs of the 
French actor network (data from \cite{actornetwork}).
 The subgraph ${\cal G}_4$ is induced by the set of 
marked vertices obtained as follows: we put a mark on each vertex $v$
with probability  $d^{-\tau}(v)$ and independently of the 
other vertices. Choosing $\tau=0.5$ we obtain a random subgraph denoted ${\cal G}_4$.  
  In our case the realized number of marked 
vertices $n_1=8871$.  ${\cal G}_5$ is the subgraph of the French 
actor network induced by $n_1$ vertices
drawn uniformly at random (without replacement) from the set of all vertices. 
Now both subgraphs have the same number of vertices, but the 
degree sequence of ${\cal G}_4$ is much more regular than that of ${\cal G}_5$.
\begin{figure}
    \begin{center}
        \includegraphics[scale=0.5]{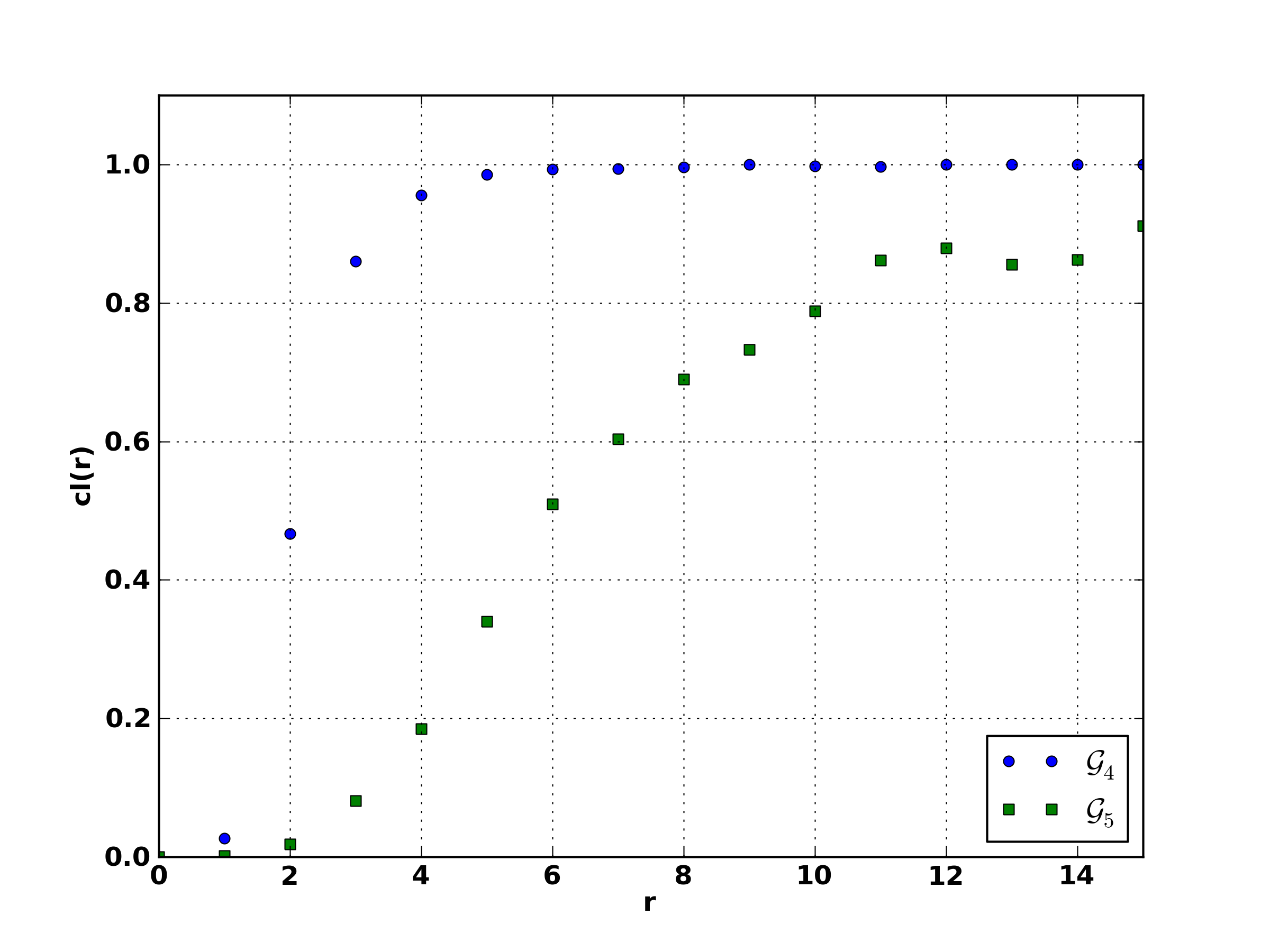}
    \end{center}
    \caption{\small Sampling a subgraph with a randomized degree constraint.}
 \label{fig.6}
\end{figure}

Finally, we examine how well a random  intersection graph fits the real data. 
For this purpose we consider a memoryless actor network obtained as follows.
Assume every  actor of a given actor graph has forgotten about the titles of movies 
he or she 
acted in and only  remembers the number of movies. 

We first  simulate 
an instance of the {\it active} memoryless graph where each actor chooses films 
independently and uniformly at random from a given set of ${\tilde m}$ films so 
that the number of 
films chosen by each actor is the same as in the true actor graph. 
In the active memoryless
graph all films have equal chances to be selected by any of actors. 
We remark that in the case where  
${\tilde m}=m$, i.e., the number of films in the active memoryless  graph is the same 
as in the real 
underlying actor network, the expected degree of the memoryless graph does not match 
the average degree of the real network. We can easily 
adjust the number of films (of the memoryless graph) so that these degrees match.
We denote this number  $m'$  and call the active memoryless graph 
with ${\tilde m}=m'$  {\it adjusted} one. 
In Figure 7 we plot clustering function (\ref{Cl1}) of two instances of  memoryless graphs 
for comparison with the underlying French actor 
network: one with the true number of films  and another with the adjusted number of 
films.

We secondly simulate an instance of the {\it inhomogeneous} memoryless graph where 
an actor $v_i$ 
chooses the 
film $w_j$ with probability $a_ib_jM^{-1}$ independently for each $i$ and $j$. 
Here the numbers $a_i, b_j$ 
are observed characteristics of the underlying actor network: $v_i$ acted in $a_i$ films;
$b_j$ actors acted in the film $w_j$.  $M=\sum_{1\le i\le n}a_i=\sum_{1\le j\le m}b_j$ 
is the total number of links
of the bipartite graph
where actors are linked to films.
In Figure 8 we plot clustering function (\ref{Cl1}) of an instance of the inhomogeneous memoryless graph
of the French actor network. Here we observe a remarkable accuracy of the approximation of the real 
clustering function by that of the memoryless graph.  We remark that in comparison with 
active memoryless graphs of Figure 7, that only use the data $a_1,\dots, a_n$, the 
inhomogeneous memoryless graph of Figure 8 uses, in addition, the numbers
$b_1,\dots, b_m$.

We remark that Theorems \ref{T1}, \ref{T2} and \ref{T3} establish a first 
order asymptotics to the 
clustering function $cl(\cdot)$ of random intersection graphs
having a square integrable asymptotic degree distribution. An interesting question 
were about a
power law random intersection graph whose asymptotic degree distribution has
infinite second 
moment:
Is there a limiting shape of the clustering function for $n\to+\infty$ in this case? 
Is there a theoretically valid approximation to the clustering function  that
explains the gradual slope of $cl(\cdot)$ of observed empirical plots?
It would also be interesting to learn about a higher order asymptotics of 
the clustering function 
$cl(\cdot)$ that refines results of Theorems \ref{T1}, \ref{T3} and could perhaps
 better explain 
the empirical data.

\begin{figure}
    \begin{center}
        \includegraphics[scale=0.5]{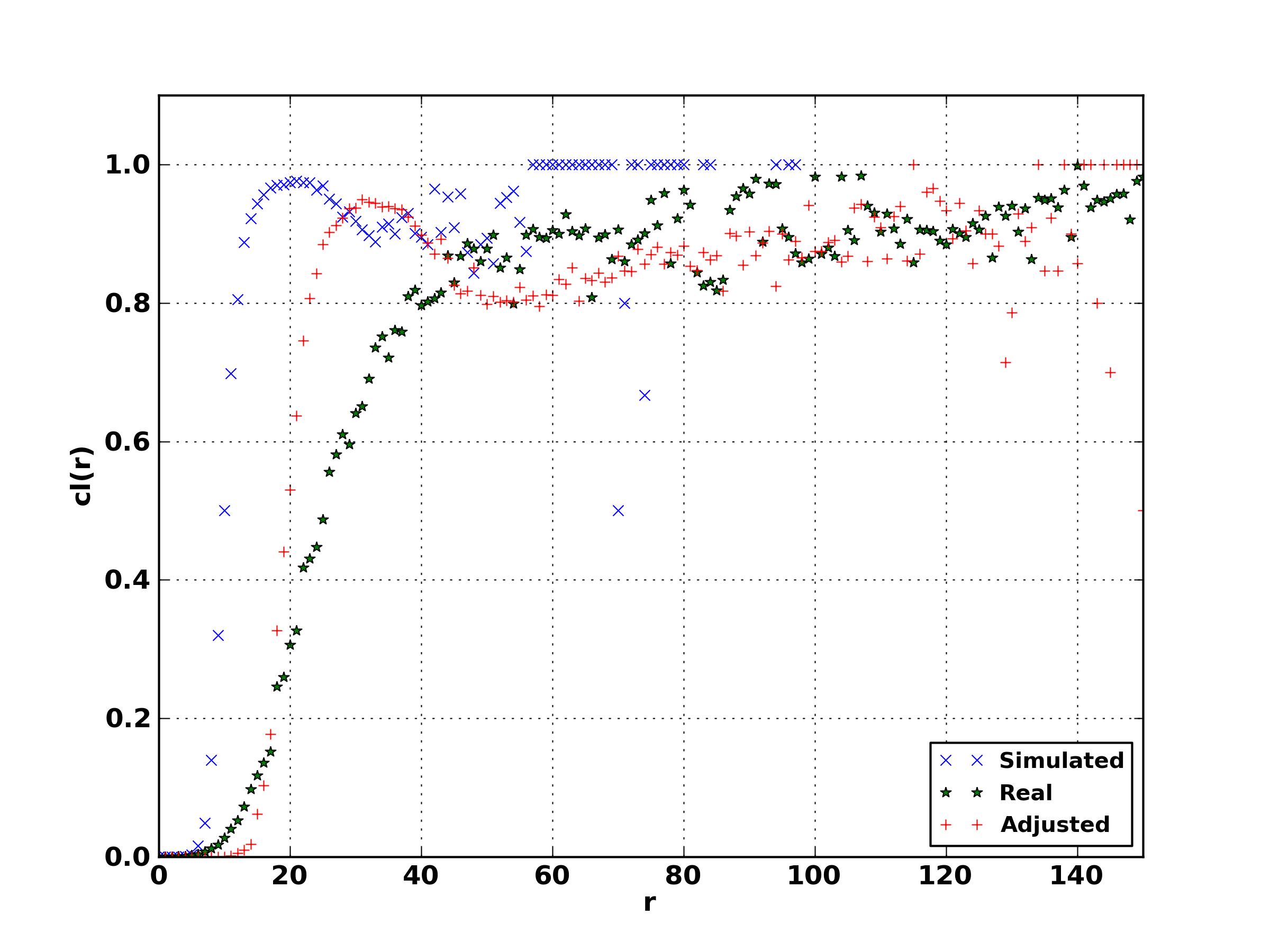}
    \end{center}
    \caption{\small     The French actor network and two simulated active memoryless networks.}
 \label{fig.7}
\end{figure}

\begin{figure}
    \begin{center}
        \includegraphics[scale=0.5]{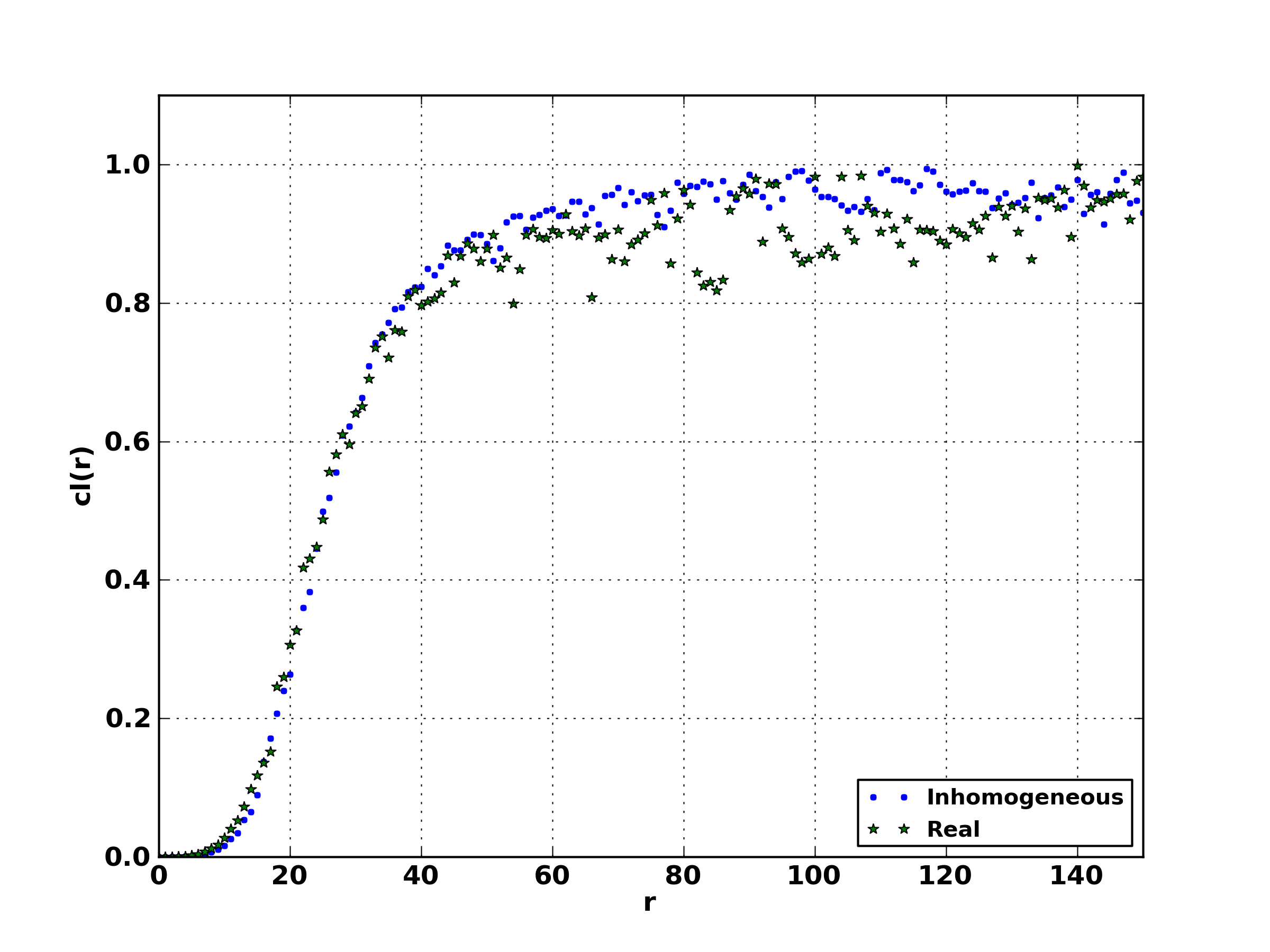}
    \end{center}
    \caption{\small     The French actor network and a simulated inhomogeneous memoryless network.}
 \label{fig.8}
\end{figure}

\section{Proofs}

The section is organized as follows: we first we formulate two auxiliary lemmas,
 then we prove Theorems \ref{T1}, \ref{T2} and \ref{T3}.

\begin{lem}\label{LeCamLemma} {\rm (See, e.g., \cite{Steele})} Let $S={\mathbb I}_1+ {\mathbb I}_2+\dots+ {\mathbb I}_n$ be the sum of independent random indicators
with probabilities $\PP({\mathbb I}_i=1)=p_i$. Let $\Lambda$ be Poisson random variable with mean $p_1+\dots+p_n$. The total variation
distance between the distributions $P_S$ of $P_{\Lambda}$ of $S$ and $\Lambda$
\begin{equation}\label{LeCam}
\sup_{A\subset \{0,1,2\dots \}}|\PP(S\in A)-\PP(\Lambda\in A)|\le \sum_{i}p_i^2.
\end{equation}
\end{lem}

\begin{lem}\label{xxx} {\rm (\cite{Bloznelis2011+})} Given integers $1\le s\le d_1\le d_2\le m$, let 
 $D_1,D_2$ be independent random subsets of the set $W=\{1,\dots, m\}$ such that
$D_1$ (respectively $D_2$) is uniformly distributed in the class of subsets of $W$ of size $d_1$ (respectively $d_2$).
The  probabilities 
$p':=\PP(|D_1\cap D_2|=s)$
and
$p'':=\PP(|D_1\cap D_2|\ge s)$ satisfy
\begin{equation}\label{xx}
\left(1-\frac{(d_1-s)(d_2-s)}{m+1-d_1}\right)p^*_{d_1,d_2,s}
\
\le 
\
p'
\
\le 
\
p'' 
\le 
\
 p^*_{d_1,d_2,s},
\end{equation}
Here we denote $p^*_{d_1,d_2,s}={\tbinom{d_1}{s}}{\tbinom{d_2}{s}} {\tbinom{m}{s}}^{-1}$.
\end{lem}

 {\bf 5.1. Active graph.}
By $X_i=X_{ni}$ we denote the size of the set $D_i$ in $G_{(n)}$. 
Furthermore, we write
$Z_i=Z_{ni}=\beta_n^{-1/2}X_{ni}$ and  put $Z_{01}:=Z$. We denote ${\bar z}_k=\E Z_{n1}^k$
and introduce the function
 \begin{equation}\label{z*2}
 t\to \varphi(t)=\sup_{n\ge 0}\E Z_{n1}^2{\mathbb I}_{\{Z_{n1}\ge t\}}.
 \end{equation}
  We remark  that conditions (i), (ii-2) 
imply $\varphi(t)=o(1)$ as $t\to+\infty$ (see, e.g.,  \cite{Bloznelis2011+}) and
  $z_{*2}:= \sup_{n\ge 0}\E Z_{n1}^2<\infty$.
By ${\tilde \PP}$ and ${\tilde \E}$  
we denote the conditional probability  and  expectation given $X_1$, $X_2$.
By ${\tilde \PP}'$ and ${\tilde \E}'$ we denote the  conditional probability  and  expectation given
 $D_1$, $D_2$. 
We introduce events ${\cal A}=\{v_1\sim v_2\}$, ${\cal A}_i=\{|D_1\cap D_2|=i\}$ and 
probabilities 
\begin{displaymath}
p_i(r)=\PP({\cal A}_i\cap \{d_{12}=r\}),
\end{displaymath}
By $f_r(\lambda)=e^{-\lambda}\lambda^r/r!$ we denote the Poisson probability.

\begin{proof}[Proof of Theorems \ref{T1} and \ref{T2}] 
We have
\begin{equation}\label{b4-3}
cl(r)
=
\PP({\cal A}|d_{12}=r)
=
\frac{\PP({\cal A}\cap \{d_{12}=r\})}{\PP(d_{12}=r)}.
\end{equation}
In order to evaluate the numerator we write ${\cal A}=\cup_{i\ge 1}{\cal A}_i$ and 
apply the total probability formula
\begin{equation}\label{A-k}
\PP({\cal A}\cap \{d_{12}=r\})
=
\sum_{i\ge 1}p_i(r)=\sum_{1\le i\le k}p_i(r)+R_k(r).
\end{equation}
Here $R_k(r)=\sum_{i> k}p_i(r)\le \PP(|D_1\cap D_2|\ge k+1)$. 
Similarly we expand the denominator of (\ref{b4-3})
\begin{equation}\label{d-12-k}
\PP(d_{12}=r)
=
\sum_{i\ge 0}p_i(r)=\sum_{0\le i\le k}p_i(r)+R_k(r).
\end{equation}

In order to prove Theorem \ref{T1} we choose  
$k=1$ in (\ref{A-k}), (\ref{d-12-k})  and invoke the asymptotic expressions of $p_i(r)$ 
and the upper bound for $ \PP(|D_1\cap D_2|\ge k+1)$ 
shown in Lemma  \ref{pagrindine}. Then, observing that  as $n\to+\infty$ we have
 ${\bar z}_k= z_k+o(1)$, for $k=1,2$, 
and $\alpha=\beta^{-1/2}z_1/z_2+o(1)$ (see (\ref{cc})), we obtain (\ref{b2-4}).

Theorem \ref{T2} is obtained in the same way, but now we choose $k=2$. 
\end{proof}


Given a sequence of 
random variables $\{Y_n\}$ and $r=0,1,\dots$
we write $Y_n\prec O_r$  to denote the fact that 
$\E |Y_n|=O(n^{-2})$, for $r=0,1$, and $\E |Y_n|=O(n^{-2}\beta_n^{-1/2})+o(n^{-2})$, 
for $r\ge 2$.

\begin{lem}\label{pagrindine} Assume that $\beta_n\to \beta\in (0,+\infty]$. 
Suppose that (i), (ii-2) hold. Denote $\Lambda_1=\beta_n^{-1/2}{\bar z}_1$ 
and $\Lambda_2={\bar z}_2-\beta_n^{-1/2}{\bar z}_1$.
We have as $n\to+\infty$
\begin{eqnarray}\label{zzz}
&&
p_0(0)=1-o(1),
\qquad
p_0(r)=o(n^{-2}),
\qquad
\
r\ge 3,
\\
\label{zzz+}
&&
p_0(r)=n^{-r}(r!)^{-1}\Lambda_2^r{\bar z}_r^2+o(n^{-r}), 
\qquad
\quad
\
\quad
r=1,2,
\\
\label{www}
&&
p_1(r)=n^{-1}{\bar z}_1^2f_r(\Lambda_1)+O_r, 
\qquad
\qquad\qquad\quad
r\ge 0,
\\
\label{ttt}
&&
p_2(r)=2^{-1}n^{-2}f_r(2\Lambda_1)\Lambda_2^2
+O_2,
\qquad
\qquad\quad
r\ge 0.
\end{eqnarray}
 Furthermore, we have
\begin{equation}\label{b4-1}
\PP(|D_1\cap D_2|\ge 3)=o(n^{-2})
\qquad
{\text{and}}
\qquad
\PP(|D_1\cap D_2|\ge k)=O(n^{-k}), 
\quad k=1,2.
 \end{equation}
\end{lem}

\begin{proof}[Proof of Lemma \ref{pagrindine}] 
Before the proof we introduce some notation and collect auxiliary inequalities. Then we give 
an outline of the proof. Afterwards we prove (\ref{zzz}), (\ref{zzz+}) 
and (\ref{www}), (\ref{ttt}).

By $c_*$ we denote a generic positive constant. 
By ${\mathbb I}_{\cal B}$ we denote the indicator of an event
 ${\cal B}$ and write ${\overline {\mathbb I}}_{\cal B}=1-{\mathbb I}_{\cal B}$. In the proof we use several indicators
\begin{eqnarray}\nonumber
&&
{\mathbb I} \  \ ={\mathbb I}_{\{X_1+X_2<\varepsilon^2n\beta_n^{1/2}\}},
\quad
{\mathbb I}_j={\mathbb I}_{\{X_j< 0.5 \varepsilon^2n\beta_n^{1/2}\}},
\quad
{\mathbb I}_{*j}={\mathbb I}_{\{X_j\le \beta_n^{1/2}\varepsilon^{-1}\}},
\\
\nonumber
&&
{\mathbb I}_{\star j}={\mathbb I}_{\{X_j\le \varepsilon m\}},
\qquad
\qquad
\
{\bf I}_j={\mathbb I}_{\{X_j<m^{1/2}n^{-1/4}\}},
\quad
{\bf I}_{*j}={\mathbb I}_{\{X_j\le 0.5 m\}}.
\end{eqnarray}
Some of them depend on  $\varepsilon>0$, value of which will be clear from the context. 
We denote 
\begin{equation}\label{g23-1}
 {\tilde q}_0=n^{-2}\Lambda_2 Z_1Z_2,
\qquad
{\tilde q}_1=n^{-1}\Lambda_1,
\qquad
{\tilde q}_2=2n^{-1}\Lambda_1,
\end{equation}
and, for $i=1,2$ we write
\begin{displaymath}
{\tilde\lambda}_i=n{\tilde q}_i,
\qquad
\lambda_i=(n-2) q_i,
\qquad
 q_i={\tilde \PP}'(v_1\sim v_3, v_2\sim v_3){\mathbb I}_{{\cal A}_i}, 
\qquad
\varkappa_i=(X_1)_i(X_2)_i/(i!(m)_i).
\end{displaymath}
We note that (\ref{xx}) implies
\begin{equation}\label{g19-1}
 \varkappa_i(1-X_1X_2/(m-X_1))
\le
{\tilde\PP}({\cal A}_i)
\le 
\varkappa_i.
\end{equation}
In particular, we have
\begin{equation}
\label{g19-2}
{\tilde \PP}({\cal A}_i)
\le 
{\tilde \PP}({\cal A}_i)({\bf I}_1{\bf I}_2+{\overline{\bf I}}_1+{\overline{\bf I}}_2)
\le 
\varkappa_i{\bf I}_1{\bf I}_2+{\overline{\bf I}}_1+{\overline{\bf I}}_2\le
n^{-i/2}+{\overline{\bf I}}_1+{\overline{\bf I}}_2.
\end{equation}
We 
will 
use the following  properties of the function
$\lambda\to f_r(\lambda)$.
For $r=0,1,\dots$, it follows from the
mean value theorem
$f_r(t)-f_r(s)=f_r'(\xi)(t-s)$, where 
$0<s\le \xi \le t$, combined with inequalities
$|f_r'(\xi)|\le 1$ and 
$|f_{2+r}'(\xi)|\le \xi$ that
\begin{equation}\label{meanvalue}
|f_r(s)-f_r(t)|\le |s-t|
\qquad
{\text{and}}
\qquad
|f_{2+r}(s)-f_{2+r}(t)|\le (s+t)|s-t|.
\end{equation}

Now we outline the proof. In order to evaluate $p_i(r)$  we write
\begin{equation}\label{pir}
p_i(r)=
\E ({\tilde \E}'({\mathbb I}_{{\cal A}_i}{\mathbb I}_{\{d_{12}=r\}}))
=
\E ({\mathbb I}_{{\cal A}_i}{\tilde \E}'{\mathbb I}_{\{d_{12}=r\}})
=
\E ({\mathbb I}_{{\cal A}_i}{\tilde \PP}'(d_{12}=r))
\end{equation}
%
and observe that, given $D_1, D_2$ satisfying $|D_1\cap D_2|=i$,  the 
random variable
\begin{displaymath}
d_{12}=\sum_{3\le j\le n}{\mathbb I}_{\{v_1\sim v_j\}}{\mathbb I}_{\{v_2\sim v_j\}}
\end{displaymath}
 has binomial distribution $\Bin(n-2,q_i)$. We first approximate 
${\tilde \PP}'(d_{12}=r)$ in (\ref{pir})
 by the Poisson probability $f_r(\lambda_i)$.
Then, we approximate $\lambda_i$ by ${\tilde\lambda}_i$, and $f_r(\lambda_i)$
by $f_r({\tilde \lambda}_i)$. 
We obtain
\begin{equation}\label{pir222}
p_i(r)
=
\E ({\mathbb I}_{{\cal A}_i}f_r({\tilde \lambda}_i))
+
\E ({\mathbb I}_{{\cal A}_i}\Delta_{r,i})
=
\E ({\tilde \PP}({\cal A}_i)f_r({\tilde \lambda}_i))
+
\E ({\mathbb I}_{{\cal A}_i}\Delta_{r,i}),
\end{equation}
where, for $|D_1\cap D_2|=i$, we denote
\begin{eqnarray}\label{g25-1}
&&
\Delta_{r,i}:={\tilde \PP}'(d_{12}=r)-f_r({\tilde \lambda}_i)=\Delta_{r,i}'+\Delta_{r,i}'',
\\
\nonumber
&&
\Delta_{r,i}'={\tilde \PP}'(d_{12}=r)-f_r(\lambda_i),
\qquad
\Delta_{r,i}''=f_r(\lambda_i)-f_r({\tilde \lambda}_i).
\end{eqnarray}
Next we show that the remainder term $\E ({\mathbb I}_{{\cal A}_i}\Delta_{r,i})$
of (\ref{pir222}) is negligible. For this purpose we 
estimate using LeCam's lemma (see Lemma \ref{LeCamLemma})
\begin{equation}\label{pirf}
 |\Delta'_{r,i}|\le nq^2_i,
 \end{equation}
and estimate $\Delta_{r,i}''$ 
combining 
(\ref{meanvalue}) with the approximations $q_i\approx {\tilde q}_i$. 
We briefly explain these approximations.
Let $\{w_1^*,\dots, w_i^*\}$ denote the intersection 
$D_1\cap D_2$ provided 
it is non empty. Denote $n_j=|D_3\cap D_j|$, $j=1,2$.
We  split
\begin{displaymath}
q_0=q_{01}+q_{02},
\qquad
q_1=q_{11}+q_{12}, 
\qquad
q_2=q_{21}+q_{22}+q_{23}+q_{24},
\end{displaymath}
where
\begin{eqnarray}\nonumber
&&
q_{01}={\tilde \PP}'(n_1=1, n_2=1){\mathbb I}_{{\cal A}_0},
\qquad \
q_{02}={\tilde \PP}'(n_1+n_2\ge 3, n_1\ge 1, n_2\ge 1){\mathbb I}_{{\cal A}_0},
\\
&&
q_{11}={\tilde \PP}'(w_1^*\in D_3){\mathbb I}_{{\cal A}_1},
\qquad
\quad
\ \ \
\quad
q_{12}={\tilde \PP}'(w_1^*\notin D_3, n_1\ge1, n_2\ge 1){\mathbb I}_{{\cal A}_1},
\\
\nonumber
&&
q_{21}={\tilde \PP}'(w_1^*\in D_3, w_2^*\notin D_3){\mathbb I}_{{\cal A}_2},
\quad
q_{22}={\tilde \PP}'(w_1^*\notin D_3, w_2^*\in D_3){\mathbb I}_{{\cal A}_2},
\\
\nonumber
&&
q_{23}={\tilde \PP}'(w_1^*,w_2^*\in D_3){\mathbb I}_{{\cal A}_2},
\qquad
\quad
\ \ 
q_{24}={\tilde \PP}'(w_1^*,w_2^*\notin D_3, n_1\ge 1, n_2\ge 1){\mathbb I}_{{\cal A}_2},
\end{eqnarray}
and approximate
  $q_0\approx q_{01}\approx {\tilde q}_0{\mathbb I}_{{\cal A}_0}$, 
$q_1\approx q_{11}={\tilde q}_1{\mathbb I}_{{\cal A}_1}$
and $q_2\approx q_{21}+q_{22}\approx {\tilde q}_2{\mathbb I}_{{\cal A}_2}$.


{\it Proof of (\ref{zzz}), (\ref{zzz+}).}  In 
order to prove (\ref{zzz}),  (\ref{zzz+}) we show that
\begin{eqnarray}\label{g13-1}
&&
\E ({\mathbb I}_{{\cal A}_0}\Delta_{r,0})
=
o(n^{-r\wedge 2}), 
\qquad
r\ge 0
\\
\label{g14-9}
&&
\E ({\mathbb I}_{{\cal A}_0} f_r({\tilde \lambda}_0))
=
 (r!)^{-1}\E{\tilde \lambda}_0^r+o(n^{-r}),
\qquad
r=0,1,2,
\\
\label{g25-5}
&&
\E  ({\mathbb I}_{{\cal A}_0}f_r({\tilde \lambda}_0))=o(n^{-2}),
\qquad
r\ge 3.
 \end{eqnarray}
We firstly  prove (\ref{g13-1}). In the case where $\beta<\infty$ we find $n_0>0$ such that
$\beta<2\beta_n$ for $n\ge n_0$. In the case where $\beta=+\infty$ we find $n_0$ such that
$\beta_n>1$ for $n\ge n_0$. 
In order to prove (\ref{g13-1}) we 
show that for any  $0<\varepsilon<\min\{0.5\beta^{1/2}, 0.1\}$ and $n\ge n_0$ we have
\begin{eqnarray}\label{g14-1}
&&
\E({\mathbb I}_{{\cal A}_0}|\Delta_{r,0}|)
\le 
c_*n^{-3}+c_*n^{-r\wedge 2}R_1(\varepsilon)+c_*n^{-2}\varepsilon^{-4}R_2(\varepsilon),
\\
\nonumber
&&
R_1(\varepsilon):=\varphi(\varepsilon^{-1})+\varepsilon+m^{-1}+n^{-1},
\qquad 
R_2(\varepsilon):=\varphi(0.5\varepsilon^2 n)(1+\varepsilon^{-4}n^{-2}).
\end{eqnarray}
We remark that
(\ref{g14-1}) combined with the relation $\lim_{t\to+\infty}\varphi(t)=0$ 
implies (\ref{g13-1}).

Let us prove (\ref{g14-1}). Given $\varepsilon$,  we write
 $\Delta_{r,0}=\Delta_{r,0}{\mathbb I}+\Delta_{r,0}{\overline{\mathbb I}}$ and show that
\begin{eqnarray}
\label{g14-3}
&&
\E |\Delta_{r,0}|{\overline{\mathbb I}}
\le
\E {\overline{\mathbb I}}
\le 
c_*n^{-2}\varepsilon^{-4}R_2(\varepsilon),
\\
\label{g14-2}
 &&
\E {\mathbb I}_{{\cal A}_0}|\Delta_{r,0}|{\mathbb I}
\le  
c_*n^{-3}+c_*n^{-r\wedge 2}R_1(\varepsilon).
\end{eqnarray}
The first inequality of (\ref{g14-3}) is obvious. In order to prove the second one 
we combine the inequalities
\begin{displaymath}
 \varepsilon^{4}n^2\E {\overline {\mathbb I}}\le \beta_n^{-1}\E(X_1+X_2)^2{\overline {\mathbb I}}
\le 2\beta_n^{-1}\E (X_1^2+X_2^2){\overline {\mathbb I}}=4\beta_n^{-1}\E X_1^2{\overline {\mathbb I}},
\end{displaymath}
 which follow from  Markov's inequality, with the inequalities
\begin{displaymath}
 \beta_n^{-1}\E X_1^2{\overline {\mathbb I}}\le 
\beta_n^{-1}\E X_1^2({\overline {\mathbb I}}_1+{\overline {\mathbb I}}_2)
\le 
\beta_n^{-1}
\E X_1^2({\overline {\mathbb I}}_1+4\varepsilon^{-4}n^{-2}\beta_n^{-1}X_2^2{\overline {\mathbb I}}_2)
\le
c_* R_2(\varepsilon).
\end{displaymath}
Here we applied the inequality 
${\overline{\mathbb I}}\le {\overline{\mathbb I}}_{1}+{\overline{\mathbb I}}_{2}$ and then Markov's 
inequality. 

In order to prove (\ref{g14-2}) 
we write $\Delta_{r,0}{\mathbb I}=\Delta_{r,0}'{\mathbb I}+\Delta_{r,0}''{\mathbb I}$, 
see (\ref{g25-1}), and invoke the inequalities
\begin{equation}\label{g27-5}
\E {\mathbb I}_{{\cal A}_0}|\Delta_{r,0}'|{\mathbb I}
\le  
c_*n^{-3}
\qquad
{\text{ and}}
\qquad 
\E {\mathbb I}_{{\cal A}_0}|\Delta_{r,0}''|{\mathbb I}
\le 
c_*n^{-r\wedge 2}R_1(\varepsilon).
\end{equation} 
The first inequality of (\ref{g27-5}) follows from  (\ref{g23-1}), (\ref{pirf}) 
 and inequalities 
$q_0^2\le 2{\tilde q}_0^2+2(q_0-{\tilde q}_0)^2$, and
\begin{equation}\label{g14-7}
 {\mathbb I}_{{\cal A}_0}|q_0-{\tilde q}_0|{\mathbb I}
\le 
c_*n^{-1}m^{-1}X_1X_2R_1(\varepsilon).
\end{equation}
The second inequality of (\ref{g27-5}) follows from (\ref{meanvalue}) and (\ref{g14-7}).
%

We complete the proof of (\ref{g13-1}) by showing (\ref{g14-7}).
  To this aim we prove that for  $D_1,D_2$ 
satisfying $|D_1\cap D_2|=0$ the following inequalities hold true
\begin{eqnarray}\label{g13-3}
 &&
(1-3\varepsilon){\tilde q}_0{\mathbb I}
-
\varphi(\varepsilon^{-1})\frac{X_1X_2}{nm}{\mathbb I}
\le 
q_{01}{\mathbb I}
\le 
(1+2\varepsilon){\tilde q}_0{\mathbb I},
\\
\label{g13-7}
 &&
 q_{02}{\mathbb I}
\le 
2n^{-1}m^{-1}X_1X_2(\varphi(\varepsilon^{-1})+2\varepsilon z_{*2}).
\end{eqnarray}
Let us prove  (\ref{g13-3}).
We write
\begin{eqnarray}
&&
q_{01}={\tilde \E}'q_{01}^*,
\qquad
{\text{where}}
\qquad
q_{01}^*={\tilde \PP}'(n_1=1, n_2=1|X_3)=\tau_1\tau_2,
\\
&&
\tau_1={\tilde \PP}'(n_1=1|X_3),
\qquad
\tau_2={\tilde \PP}'(n_2=1|n_1=1,X_3),
\end{eqnarray}
and apply  (\ref{xx}) to probabilities $\tau_1$ and $\tau_2$. We obtain
\begin{displaymath}
\frac{X_1X_2(X_3)_2}{m^2}(1-\theta_1)(1-\theta_2)
\le
\tau_1\tau_2
\le 
\frac{X_1X_2(X_3)_2}{m^2}\theta_3.
\end{displaymath}
Here $\theta_1=\frac{(X_1-1)(X_3-1)}{m-X_1+1}$, $\theta_2=\frac{(X_2-1)(X_3-2)}{m-X_1-X_2+1}$ and $\theta_3=\frac{m}{m-X_1}$.
Next, we observe that, by our choice of $\varepsilon$, we have 
$\varepsilon\le \beta_n^{1/2}$ for $n\ge n_0$. In particular, the 
inequality $X_1+X_2\le \varepsilon^2n\beta_n^{1/2}$
implies $X_1+X_2\le \varepsilon m$. Assuming, in addition, that 
$X_3\le \varepsilon^{-1}\beta_n^{1/2}$, we obtain
$X_iX_3\le \varepsilon m$, for $i=1,2$. These inequalities imply 
$\theta_i\le \varepsilon/(1-\varepsilon)$, $i=1,2$, and $\theta_3\le 1+2\varepsilon$. 
Note that $\varepsilon<0.1$.
Hence, we have
\begin{displaymath}
(1-\theta_1)(1-\theta_2)\ge 1-\theta_1-\theta_2\ge 1-3\varepsilon.
\end{displaymath}
Now, we write
\begin{displaymath}
\frac{X_1X_2(X_3)_2}{m^2}(1-3\varepsilon){\mathbb I}_{*3}{\mathbb I} 
\le
{\mathbb I}_{*3}{\mathbb I}\tau_1\tau_2
\le
\tau_1\tau_2{\mathbb  I}
\le 
\frac{X_1X_2(X_3)_2}{m^2}(1+2\varepsilon){\mathbb I}
\end{displaymath}
and, using  identities 
${\mathbb I}q_{01}
=
{\mathbb I}{\tilde \E}'q_{01}^*
=
{\tilde \E}'{\mathbb I}q_{01}^*
=
{\tilde \E}'{\mathbb I}\tau_1\tau_2$, 
we obtain
\begin{eqnarray}\nonumber
&&
{\mathbb I}q_{01}\le {\mathbb I}(1+2\varepsilon)\frac{X_1X_2}{m^2}{\tilde \E}'(X_3)_2
=
{\mathbb I}(1+2\varepsilon){\tilde q}_0,
\\
\nonumber
&&
{\mathbb I}q_{01}
\ge
{\mathbb I}(1-3\varepsilon)\frac{X_1X_2}{m^2}{\tilde \E}'\bigl({\mathbb I}_{*3}(X_3)_2\bigr)
\ge 
{\mathbb I}(1-3\varepsilon){\tilde q}_0
-
{\mathbb I}\varphi(\varepsilon^{-1})
\frac{X_1X_2}{nm}.
\end{eqnarray}
In the last step we used the inequalities 
$\beta_n^{-1}{\tilde \E}'(X_3)_2(1-{\mathbb I}_{*3})
\le 
\beta_n^{-1}{\tilde \E}'X_3^2(1-{\mathbb I}_{*3})
\le 
\varphi(\varepsilon^{-1})$.

Now we prove (\ref{g13-7}). To this aim we write 
\begin{displaymath}
 q_{02}\le q_{03}+q_{04},
\qquad
q_{03}:={\tilde \PP}'(n_1\ge 2, n_2\ge 1){\mathbb I}_{{\cal A}_0},
\qquad
q_{04}:={\tilde \PP}'(n_1\ge 1, n_2\ge 2){\mathbb I}_{{\cal A}_0}
\end{displaymath}
and show that  for  $D_1,D_2$ 
satisfying $|D_1\cap D_2|=0$ the following inequalities hold true
\begin{equation}\label{2013-02-16}
 q_{0j}{\mathbb I}
\le
n^{-1}m^{-1}X_1X_2
\bigl(4\varepsilon {\bar z}_2+2\varphi(\varepsilon^{-1})\bigr).
\qquad
 j=3,4.
\end{equation}
We only prove (\ref{2013-02-16}) for $j=4$ (both cases $j=3,4$ are identical). 
Observing that  probabilities  $p_{k*}:={\tilde \PP}'(n_1\ge 1, n_2\ge k|X_3)$ 
satisfy the inequality $p_{2*}\le p_{1*}$, we  write
\begin{equation}\label{g13-4}
q_{04}
=
{\tilde \E}'p_{2*}
=
{\tilde \E}'p_{2*}({\mathbb I}_{*3}+{\overline {\mathbb I}}_{*3})
\le 
{\tilde \E}'p_{2*}{\mathbb I}_{*3}+{\tilde \E}'p_{1*}{\overline {\mathbb I}}_{*3}.
\end{equation}
Next, we split 
\begin{displaymath}
p_{k*}=\tau_{k*}\tau_{*},
\qquad
\tau_{*}={\tilde \PP}'(n_1\ge 1|X_3),
\qquad
\tau_{k*}={\tilde \PP}'(n_k\ge k|n_1\ge 1, X_3)
\end{displaymath}
and apply (\ref{xx}) to the probabilities $\tau_{*}$ and $\tau_{k*}$. We have
\begin{displaymath}
\tau_{*}\le \frac{X_1X_3}{m},
\quad
\tau_{1*}\le \frac{X_2(X_3-1)}{m-X_1}
\le
\frac{X_2X_3}{m}\theta_3,
\quad
\tau_{2*}
\le \frac{X_2^2X_3^2}{(m-X_1)^2}
\le 
\frac{X_2^2X_3^2}{m^2}\theta_3^2.
\end{displaymath}
We recall that $\theta_3=m/(m-X_1)$ satisfies 
$\theta_3{\mathbb I}\le (1+2\varepsilon){\mathbb I}<2{\mathbb I}$.
Collecting these inequalities in (\ref{g13-4}) 
we obtain (\ref{2013-02-16}):
\begin{eqnarray}
\nonumber
 q_{04}{\mathbb I}
&
\le
& 
4\frac{X_1X_2^2}{m^3}{\mathbb I}{\tilde \E}'X_3^3{\mathbb I}_{*3}
+
2\frac{X_1X_2}{m^2}{\mathbb I}{\tilde \E}' X_3^2{\overline {\mathbb I}}_{*3}
\\
\nonumber
&
\le
&
4\varepsilon \frac{X_1X_2}{nm}{\bar z}_2+ 2\varphi(\varepsilon^{-1})\frac{X_1X_2}{nm}.
\end{eqnarray}
In the last step  we used identity 
$m^{-1}{\tilde \E}' X_3^2{\overline {\mathbb I}}_{*3}=n^{-1}\varphi(\varepsilon^{-1})$ and
inequalities
\begin{displaymath}
 X_1X_2^2X_3^3{\mathbb I}_{*3}{\mathbb I}
\le 
\varepsilon m X_1X_2X_3^2.
\end{displaymath}


We secondly prove (\ref{g14-9}).
Denote
\begin{equation}\label{g14-10}
R_{01}
=
f_r({\tilde \lambda}_0)
-
(r!)^{-1}{\tilde \lambda}_0^r,
\qquad
R_{02}=({\tilde \PP}({\cal A}_0)-1) (r!)^{-1}{\tilde \lambda}_0^r.
 \end{equation}
We observe that  $1-e^{-{\tilde\lambda}_0}\le {\tilde \lambda}_0$ implies
  $|R_{01}|\le {\tilde\lambda}_0^{r+1}$. 
Furthermore, from the inequality
\begin{equation}\label{A0}
 1-{\tilde \PP}({\cal A}_0)={\tilde \PP}(D_1\cap D_2\not=\emptyset)\le X_1X_2m^{-1},
\end{equation}
see (\ref{xx}), we obtain $|R_{02}|\le {\tilde \lambda}_0^rX_1X_2m^{-1}$.
We remark that, for $r=0,1$, relation
 (\ref{g14-9}) 
follows from  the  bounds
$\E {\tilde\lambda}_0^{r+1}=O(n^{-r-1})$ and
$\E {\tilde \lambda}_0^rX_1X_2m^{-1}=O(n^{-r-1})$. Indeed, we have
\begin{displaymath}
 \E{\mathbb I}_{{\cal A}_0}f_r({\tilde \lambda}_0)-\E(r!)^{-1}{\tilde \lambda}_0^r
=
\E{\mathbb I}_{{\cal A}_0}R_{01}
+
\E({\mathbb I}_{{\cal A}_0}-1)(r!)^{-1}{\tilde \lambda}_0^r
=
\E{\mathbb I}_{{\cal A}_0}R_{01}
+
\E R_{02}=O(n^{-r-1}).
\end{displaymath}

In the case where $r=2$ we  invoke the truncation argument.
Denote
\begin{equation}\label{g20-1}
 R_{03}
=
f_2({\tilde\lambda}_0)(1-{\bf I}_1{\bf I}_2),
\qquad
R_{04}
=
 (2!)^{-1}{\tilde\lambda}_0^2(1-{\bf I}_1{\bf I}_2).
\end{equation}
We observe that inequalities 
\begin{equation}\label{III}
{\bf I}_1{\bf I}_2
\le 
1
\le 
{\bf I}_1{\bf I}_2+{\overline {\bf I}}_1+{\overline {\bf I}}_2
\end{equation} 
imply, for $j=3,4$,
\begin{equation}\label{2013-03-09++X}
 \E |R_{0j}|
\le
\E {\tilde\lambda}_0^2({\overline {\bf I}}_1+{\overline {\bf I}}_2)
\le
c_*n^{-2}\varphi(n^{1/4})=o(n^{-2}).
\end{equation}
Finally, we obtain (\ref{g14-9}) from the identities
\begin{eqnarray}\nonumber
 &&
\E {\mathbb I}_{{\cal A}_0}f_2({\tilde \lambda}_0)-\E (2!)^{-1}{\tilde \lambda}_0^2
=
\E {\mathbb I}_{{\cal A}_0}{\bf I}_1{\bf I}_2f_2({\tilde \lambda}_0)
-
\E {\bf I}_1{\bf I}_2(2!)^{-1}{\tilde \lambda}_0^2
+
\E{\mathbb I}_{{\cal A}_0}R_{03}-\E R_{04},
\\
\nonumber
&&
\E {\mathbb I}_{{\cal A}_0}{\bf I}_1{\bf I}_2f_2({\tilde \lambda}_0)
-
\E {\bf I}_1{\bf I}_2 (2!)^{-1}{\tilde \lambda}_0^2
=
\E {\mathbb I}_{{\cal A}_0}{\bf I}_1{\bf I}_2R_{01}
+
\E {\bf I}_1{\bf I}_2R_{02}
\end{eqnarray}
combined with  bounds (\ref{2013-03-09++X}) and 
\begin{displaymath}
\E {\bf I}_1{\bf I}_2(|R_{01}|+|R_{02}|)
\le 
\E {\bf I}_1{\bf I}_2({\tilde\lambda}_0^{3}+{\tilde\lambda}_0^{2}X_1X_2m^{-1})
\le c_*n^{-5/2}. 
\end{displaymath}


Let us prove (\ref{g25-5}). We write
\begin{displaymath}
 \E {\mathbb  I}_{{\cal A}_0}f_r({\tilde\lambda}_0)
\le 
\E f_r({\tilde\lambda}_0)
\le 
\E f_r({\tilde\lambda}_0)({\bf I}_1{\bf I}_2+{\overline {\bf I}}_1+{\overline {\bf I}}_2)
\end{displaymath}
and  
apply the inequalities $f_r(t)\le t^jf_{r-j}(t)\le t^j$, $0\le j\le r$. For $r\ge 3$ we obtain
\begin{eqnarray}\nonumber
&&
\E f_r({\tilde\lambda}_0){\bf I}_1{\bf I}_2
\le 
\E {\tilde \lambda}_0^3{\bf I}_1{\bf I}_2
\le 
c_*n^{-1/2}\E {\tilde \lambda}_0^2=O(n^{-5/2}),
\\
\nonumber
&&
\E f_r({\tilde\lambda}_0)({\overline {\bf I}}_1+{\overline {\bf I}}_2)
\le
\E {\tilde \lambda}_0^2({\overline {\bf I}}_1+{\overline {\bf I}}_2)
\le 
c_*n^{-2}\varphi(n^{1/4})=o(n^{-2}).
\end{eqnarray}













{\it Proof of (\ref{www}), (\ref{ttt}).} 
We remark that (\ref{www}), (\ref{ttt}) follows from (\ref{pir222}) and the bounds, 
for $i=1,2$,
\begin{eqnarray}\label{g18-1}
&&
{\mathbb I}_{{\cal A}_i}\Delta_{r,i}\prec O_{r},
\\
&&
\label{g18-2}
({\tilde \PP}({\cal A}_i)-\varkappa_i)f_r({\tilde \lambda}_i)\prec O_{r\vee i}.
\end{eqnarray}

We first prove (\ref{g18-1}). For this purpose we combine identities
\begin{displaymath}
\Delta_{r,i}
=
\Delta_{r,i}{\bf I}_{*1}+\Delta_{r,i}{\overline {\bf I}}_{*1}
=
\Delta_{r,i}'{\bf I}_{*1}+\Delta_{r,i}''{\bf I}_{*1}+\Delta_{r,i}{\overline {\bf I}}_{*1}
\end{displaymath}
with the  bounds, which are shown below,
\begin{equation}\label{g27-1}
{\mathbb I}_{{\cal A}_i} \Delta_{r,i}'{\bf I}_{*1}
\prec O_{r}, 
\qquad
{\mathbb I}_{{\cal A}_i} \Delta_{r,i}''{\bf I}_{*1}
\prec O_{r}, 
\qquad
{\mathbb I}_{{\cal A}_i} \Delta_{r,i}{\overline {\bf I}}_{*1}
\prec O_2.
\end{equation}
We remark that  the third bound  of (\ref{g27-1}) is an  easy consequence of Markov's inequality, 
\begin{displaymath}
\E{\mathbb I}_{{\cal A}_i}\Delta_{r,i}{\overline {\bf I}}_{*1}
\le 
\E {\overline {\bf I}}_{*1}
\le 
4(nm)^{-1}\varphi(0.5\sqrt{nm})=o(n^{-2}).
\end{displaymath}
Now we prove the first and  second bound of (\ref{g27-1}) in the case where $i=1$.
In the proof we use the simple identity $q_{11}={\tilde q}_1$  and inequality
\begin{equation}\label{g19-8}
 q_{12}{\bf I}_{*1}\le 2n^{-1}m^{-1}{\bar z}_2X_1X_2
\end{equation}
which hold whenever conditions of event ${\cal A}_1$ are satisfied.
We note that (\ref{g19-8}) follows from identities
\begin{eqnarray}\nonumber
q_{12}
&=&
{\tilde\E}'({\tilde \PP}'(w_1^*\notin D_3,\, n_1\ge 1, n_2\ge 1|X_3))
={\tilde \E}'({\tilde \PP}'(w_1^*\notin D_3|X_3)\tau_1'\tau_2'),
\\
\nonumber
\tau_1':
&=&
{\tilde \PP}'(n_1\ge 1|w_1^*\notin D_3, X_3),
\qquad
\tau_2':
=
{\tilde \PP}'(n_2\ge 1|n_1\ge 1, w_1^*\notin D_3, X_3)
\end{eqnarray}
and inequalities, see (\ref{xx}),
\begin{displaymath}
 \tau_1'\le (m-1)^{-1}(X_1-1)X_3,
\qquad 
\tau_2'\le(m-X_1)^{-1}(X_2-1)(X_3-1).
\end{displaymath}

Let us prove the first  bound of (\ref{g27-1}). 
Combining  (\ref{pirf}) with  inequality
$q_1^2\le 2q_{11}^2+2q_{12}^2$ 
we write
$|\Delta_{r,1}'|\le nq_1^2\le 2n{\tilde q}_1^2+2nq_{12}^2$. Hence, we obtain
\begin{displaymath}
 {\tilde \E}{\mathbb I}_{{\cal A}_1} |\Delta_{r,1}'|{\bf I}_{*1}
\le 
{\tilde\PP}({\cal A}_1)2n{\tilde q}_1^2+{\tilde \E}2nq_{12}^2{\bf I}_{*1}.
\end{displaymath}
Furthermore, invoking inequality
$\E n{\tilde q}_1^2{\tilde \PP}({\cal A}_1)\le c_*n^{-2}\beta_n^{-1}$, 
which follows from (\ref{g19-1}), and  bound
 $\E nq_{12}^2{\bf I}_{*1}=O(n^{-3})$, 
which follows from (\ref{g19-8}),
we obtain the first bound of (\ref{g27-1}).

Let us prove the second  bound of (\ref{g27-1}).
In the proof we  
use the inequalities 
\begin{equation}\label{kovo14}
|\lambda_1-{\tilde \lambda}_1|\le 2{\tilde q}_1+nq_{12},
\qquad
\lambda_1+{\tilde \lambda}_1\le 2n{\tilde q}_1+nq_{12}.
\end{equation}
For $r=0,1$ we apply (\ref{meanvalue}) and obtain 
\begin{equation}\label{2013-03-11}
 \E {\mathbb I}_{{\cal A}_1}|\Delta_{r,1}''|{\bf I}_{*1}
\le 
\E {\mathbb I}_{{\cal A}_1}|\lambda_1-{\tilde \lambda}_1| {\bf I}_{*1}
\le
2{\tilde q}_1\E {\mathbb I}_{{\cal A}_1}+n\E {\mathbb I}_{{\cal A}_1}q_{12}{\bf I}_{*1}
\end{equation}
Then we invoke the bounds $\E {\mathbb I}_{{\cal A}_1}=\E{\tilde \PP}({\cal A}_1)=O(n^{-1})$, 
see (\ref{g19-1}), and 
\begin{equation}\label{2013-03-11+}
 n\E {\mathbb I}_{{\cal A}_1}q_{12}{\bf I}_{*1}
\le 2{\bar z}_2m^{-1}\E {\mathbb I}_{{\cal A}_1}X_1X_2
= 2{\bar z}_2m^{-1}\E {\tilde \PP}({\cal A}_1)X_1X_2
=O(n^{-2}),
\end{equation}
see (\ref{g19-1}), (\ref{g19-8}). Clearly, (\ref{2013-03-11}), (\ref{2013-03-11+})
imply the second  bound of (\ref{g27-1}).

 For $r\ge 2$ we derive the second  bound of (\ref{g27-1}) from   inequalities, 
see (\ref{meanvalue}), (\ref{kovo14}),
\begin{equation}\nonumber
 {\mathbb I}_{{\cal A}_1}|\Delta_{r,1}''|
\le 
{\mathbb I}_{{\cal A}_1} |\lambda_1-{\tilde \lambda}_1|
(\lambda_1+{\tilde \lambda}_1)
\le 
{\mathbb I}_{{\cal A}_1}(4n{\tilde q}_1^2+3n^2{\tilde q}_1q_{12}+n^2q_{12}^2)
\end{equation}
combined with   relations
\begin{eqnarray}\label{03-11+A}
 &&
\E{\mathbb I}_{{\cal A}_1}n{\tilde q}_1^2=\E n{\tilde q}_1^2{\tilde \PP}({\cal A}_1)
=O(n^{-2}\beta_n^{-1}),
\\
\label{03-11+B}
&&
\E {\mathbb I}_{{\cal A}_1} n^2{\tilde q}_1q_{12}{\bf I}_{*1}
=
n^2{\tilde q}_1\E {\mathbb I}_{{\cal A}_1}q_{12}{\bf I}_{*1}
=
O(n^{-2}\beta_n^{-1/2}),
\\
\label{03-11+C}
&&
\E {\mathbb I}_{{\cal A}_1} n^2q_{12}^2{\bf I}_{*1}
\le 
4{\bar z}_2^2m^{-2}\E {\mathbb I}_{{\cal A}_1}X_1^2X_2^2
\le 
c_*n^{-5/2}+c_*n^{-2}\varphi(n^{1/4}).
\end{eqnarray}
Here (\ref{03-11+A}) follows from (\ref{g19-1}). (\ref{03-11+B}) follows from 
(\ref{2013-03-11+}). The first inequality of (\ref{03-11+C}) follows from 
(\ref{g19-8}). To show the second inequality of (\ref{03-11+C}) we invoke (\ref{g19-2})
and
write
\begin{displaymath}
\E {\mathbb I}_{{\cal A}_1}X_1^2X_2^2
=
\E {\tilde \PP}({\cal A}_1)X_1^2X_2^2
\le 
\E(n^{-1/2}+{\overline {\bf I}}_1+{\overline {\bf I}}_2)X_1^2X_2^2
\le
c_*\beta_n^2(n^{-1/2}+\varphi(n^{1/4})).
\end{displaymath}


Now we establish the first two bounds of (\ref{g27-1}) for $i=2$.
In the proof we use the  relations 
\begin{eqnarray}\label{g21-7}
q_{24}{\bf I}_{*1}
&
\le
&
c_*(nm)^{-1}X_1X_2,
\\
\label{g21-6}
q_{23}
&
=
&
{\tilde \E}(X_3)_2(m)_2^{-1}\le c_*n^{-2},
\\
\label{g28-1}
q_{21}
&
=
&
q_{22}= 2^{-1}{\tilde q}_2+R_{*}, 
\end{eqnarray}
where $|R_{*}|\le c_*n^{-2}(\beta_n^{-3/2}+\beta_n^{-1})$.
Here (\ref{g21-7}) is obtained in the same way as (\ref{g19-8}) above, and  (\ref{g21-6})
follows from (\ref{xx}). Furthemore, the first identity of (\ref{g28-1}) is obvious and  
second one is obtained from the identities
\begin{displaymath}
 q_{22}
=
{\tilde \E}'{\tilde \PP}_2'(w_2^*\in D_3|w_1^*\notin D_3, X_3)
{\tilde \PP}'(w_1^*\notin D_3| X_3)
={\tilde \E}'\frac{X_3}{m-1}\left(1-\frac{X_3}{m}\right)= \frac{\E X_3}{m}+R_{*}.
\end{displaymath}
To prove the first bound of (\ref{g27-1}) for $i=2$ we write, see (\ref{pirf}),
\begin{displaymath}
 {\mathbb I}_{{\cal A}_2}|\Delta_{r,2}'|
\le
{\mathbb I}_{{\cal A}_2}nq_2^2
\le 
2n{\mathbb I}_{{\cal A}_2}(q_{21}^2+q_{22}^2+q_{23}^2+q_{24}^2)
\end{displaymath}
and invoke the bounds, which follow from (\ref{g21-7}), (\ref{g21-6}),
(\ref{g28-1}),
\begin{displaymath}
\E q_{23}^2\le c_*n^{-4},
\qquad
\E q_{24}^2{\bf I}_{*1}\le c_*n^{-4},
\qquad
\E q_{21}^2{\mathbb I}_{{\cal A}_2}
=
q_{21}^2\E {\mathbb I}_{{\cal A}_2} \le c_*n^{-4}\beta_n^{-1}.
\end{displaymath}
In the last step we used inequalities
$\E {\mathbb I}_{{\cal A}_2}=\E {\tilde \PP}({\cal A}_2)\le \E \varkappa_2\le c_*n^{-2}$, 
see (\ref{g19-1}).

The second bound of (\ref{g27-1}) for $i=2$ follows from the relations shown below
\begin{eqnarray}\label{kovo15+3}
&&|\Delta_{r,2}''|
\le
|\lambda_2-{\tilde\lambda}_2|\le n(q_{23}+q_{24}+2|R_{*}|)+2{\tilde q}_2,
\\
\label{kovo15+1}
&&
 \E {\mathbb I}_{{\cal A}_2}(q_{23}+2|R_{*}|+2{\tilde q}_2)
=
(q_{23}+2|R_{*}|+2{\tilde q}_2)\E {\mathbb I}_{{\cal A}_2}
\le 
c_*(n^{-4}+n^{-3}\beta_n^{-1/2}),
\\
\label{kovo15+2}
&&
\E {\mathbb I}_{{\cal A}_2}q_{24}{\bf I}_{*1}
\le 
c_*(nm)^{-1}\E {\mathbb I}_{{\cal A}_2}X_1X_2
=
c_*(nm)^{-1}\E {\tilde \PP}({\cal A}_2)X_1X_2
=o(n^{-3}).
\end{eqnarray}
Here the first inequality of (\ref{kovo15+3}) follows from (\ref{meanvalue}), and the second 
inequality 
follows from (\ref{g28-1}) and the identity
\begin{displaymath}
 \lambda_2-{\tilde \lambda}_2
=
(n-2)(q_{23}+q_{24}+(q_{21}+q_{22}-{\tilde q}_2))-2{\tilde q}_2.
\end{displaymath}
Furthermore, (\ref{kovo15+1}) follows from (\ref{g21-6}), (\ref{g28-1}) and inequality
$\E {\mathbb I}_{{\cal A}_2}\le c_*n^{-2}$. Finally, the first inequality
of (\ref{kovo15+2}) follows from (\ref{g21-7}), 
and in the the last step of (\ref{kovo15+2}) we use the inequality 
${\tilde \PP}({\cal A}_2)\le 
\varkappa_2^{1/2}(n^{1/2}+{\overline {\bf I}}_1+{\overline{\bf I}}_2)$, which follows from 
(\ref{g19-2}).

Now we prove (\ref{g18-2}). Since $f_r({\tilde \lambda}_i)\le 1$ it suffices to show
that $\varkappa_i-{\tilde P}({\cal A}_i)\prec O_{r\vee i}$.
For $i=1$ we write, see (\ref{g19-1}),
\begin{equation}\label{kovo15+5}
 \varkappa_1\ge {\tilde \PP}({\cal A}_1)
\ge 
{\tilde \PP}({\cal A}_1){\bf I}_{*1}
\ge 
\varkappa_1\left(1-\frac{X_1X_2}{m-X_1}\right){\bf I}_{*1}
=
\varkappa_1-R_{11}-R_{12},
\end{equation}
where
\begin{equation}
 R_{11}=\varkappa_1 {\bar {\bf I}}_{*1} 
\le 
2\varkappa_1\frac{X_1}{m},
\qquad
R_{12}=\varkappa_1\frac{X_1X_2}{m-X_1}{\bar{\bf I}}_{*1}
\le 
2\varkappa_1\frac{X_1X_2}{m}.
\end{equation}
Hence, we have $\E|\varkappa_1-{\tilde \PP}({\cal A}_1)|\le \E|R_{11}|+\E|R_{12}|=O(n^{-2})$.

For $i=2$ we proceed as follows. 
Given $0<\varepsilon<1$ we write, see (\ref{g19-1}),
\begin{equation}\label{g27-3}
 \varkappa_2\ge {\tilde \PP}({\cal A}_2)
\ge 
{\tilde \PP}({\cal A}_2){\mathbb I}_{*1}{\mathbb I}_{*2}
\ge 
\varkappa_2\left(1-\frac{X_1X_2}{m-X_1}\right){\mathbb I}_{*1}{\mathbb I}_{*2}
=
\varkappa_2-R_{12}-R_{22},
\end{equation}
where
\begin{equation}
 R_{12}=\varkappa_2(1-{\mathbb I}_{*1}{\mathbb I}_{*2})
\le 
\varkappa_2({\bar {\mathbb I}}_{*1}+{\bar {\mathbb I}}_{*2}),
\qquad
R_{22}=\varkappa_2\frac{X_1X_2}{m-X_1}{\mathbb I}_{*1}{\mathbb I}_{*2}
\le 
\varkappa_2\frac{\beta_n\varepsilon^{-2}}{m-\beta_n\varepsilon^{-1}}.
\end{equation}
We let $\varepsilon=\varepsilon_n\to 0$ slowly enough to get 
$\frac{\beta_n\varepsilon^{-2}}{m-\beta_n\varepsilon^{-1}}=o(1)$.
Then we obtain $\E|\varkappa_2-{\tilde \PP}({\cal A}_2)|\le \E|R_{21}|+\E|R_{22}|=o(n^{-2})$

{\it Proof of (\ref{b4-1}).}  We write, for short, 
${\bar p}_k:={\tilde \PP}(|D_1\cap D_2|\ge k)$.
We have, see (\ref{xx}),
\begin{equation}\label{b4-2}
 {\bar p}_k
\le 
\varkappa_k\le Z_1^kZ_2^kn^{-k}.
\end{equation}
Taking the expected values in (\ref{b4-2}) we obtain   (\ref{b4-1})  for $k=1,2$. 
For $k=3$ we apply (\ref{g19-2}) and  write
\begin{displaymath}
{\bar p}_3={\bar p}_3^{2/3}{\bar p}_3^{1/3}
\le 
\varkappa_3^{2/3}(n^{-3/2}+{\overline{\bf I}}_1+{\overline{\bf I}}_2)^{1/3}
\le 
\varkappa_3^{2/3}(n^{-1/2}+{\overline{\bf I}}_1+{\overline{\bf I}}_2). 
\end{displaymath}
Hence, we have
$\PP(|D_1\cap D_2|\ge 3)\le \E\varkappa_3^{2/3}(n^{-1/2}+{\overline{\bf I}}_1+{\overline{\bf I}}_2)=o(n^{-2})$.
\end{proof}

\newpage
{\bf 5.2. Inhomogeneous graph.}
Before the proof of Theorem \ref{T3} we introduce some notation and show
(\ref{p-e-inh}). By $\PP^*$ and $\E^*$ we denote the conditional probability and 
expectation given $A_1,A_2, B_m$.
 Given $i,j, l\in [m]$ and $k,s,t\in [n]$, denote
\begin{eqnarray}\nonumber
&&
{\bf I}_{ki}={\mathbb I}_{\{A_kB_i\le \sqrt{mn}\}},
\qquad
{\mathbb I}_{ki}={\mathbb I}_{\{w_i\in D_k\}},
\\
\nonumber
&&
{\cal H}_i=\{w_i\in D_1\cap D_2\},
\qquad
{\cal H}_{ijk}=\{w_i\in D_1\cap D_k,\, w_j\in D_2\cap D_k\},
\\
\nonumber
&&
{\cal H}^*_i=\{w_i= D_1\cap D_2\},
\qquad
{\cal H}_{ij}=\{w_i,w_j\in D_1\cap D_2\},
\\
\nonumber
&&
{\cal L}_{i}=\{w_i\in D_1\cap D_2\cap D_3\},
\qquad
{\cal L}_{ijl}=\{w_i\in D_1\cap D_2,\, w_j\in D_1\cap D_3,\, w_l\in D_2\cap D_3\},
\\
\nonumber
&&
{\cal U}_{s}=\{v_1\not\sim v_2, v_s\sim v_1,  v_s\sim v_2\},
\qquad
{\cal U}_{st}={\cal U}_{s}\cap {\cal U}_t,
\\
\nonumber
&&
{\cal H}^*=\cup_{i\in [m]}{\cal H}^*_i,
\qquad
{\cal H}^{**}=\cup_{\{i,j\}\subset[m]}{\cal H}_{ij},
\qquad
{\cal H}^{***}=\cup_{k=3}^n\cup_{\{i,j\}\subset[m-1]}({\cal H}_{ijk}\cup{\cal H}_{jik}),
\\
\nonumber
&&
{\cal L}=\cup_{i\in[m]}{\cal L}_i,
\qquad
{\cal L}^*=\cup_{\{i,j,l\}\subset [m]}{\cal L}_{ijl},
\qquad
{\cal L}^{**}=\cup_{1\le i,j\le m, \, i\not=j}{\cal H}_{ij3}.
 \end{eqnarray}
Introduce the random variable ${\mathbb S}=\sum_{3\le k\le n}{\mathbb I}_{km}$
 and probability ${\tilde p}_r=\PP({\cal H}^*_m\cap\{{\mathbb S}=r\})$.

Let us prove (\ref{p-e-inh}).
It follows from  identity $\{v_1\sim v_2\}=\cup_{i\in [m]}{\cal H}_i$, 
by inclusion-exclusion, that
\begin{equation}\nonumber
 \sum_{i\in [m]}\PP({\cal H}_i)
-
\sum_{\{i,j\}\subset[m]}\PP({\cal H}_i\cap {\cal H}_j)
\le 
\PP(v_1\not\sim v_2)
\le 
\sum_{i\in [m]}\PP({\cal H}_i).
\end{equation}
We derive (\ref{p-e-inh}) from these inequalities using relations
\begin{equation}
 \PP({\cal H}_i)=(nm)^{-1}(a_1^2b_2+o(1)),
\qquad
\PP({\cal H}_i\cap {\cal H}_j)\le \E p_{1i}p_{2i}p_{1j}p_{2j}\le (nm)^{-2}a_2^2b_2^2.
\end{equation}
To show the first relation we apply the inequality 
${\bf I}_{1i}{\bf I}_{2i}\ge 1-{\bar{\bf I}}_{1i}-{\bar{\bf I}}_{2i}$
and write
\begin{equation}\label{sest-06-1}
 \frac{A_1A_2B_i^2}{nm}
\ge 
p_{1i}p_{2i}
\ge 
p_{1i}p_{2i}{\bf I}_{1i}{\bf I}_{2i}
=
\frac{A_1A_2B_i^2}{nm}{\bf I}_{1i}{\bf I}_{2i}
\ge
\frac{A_1A_2B_i^2}{nm}(1-{\bar{\bf I}}_{1i}-{\bar{\bf I}}_{2i}).
\end{equation}
Then we take  the expected values in (\ref{sest-06-1}), use the identity 
$\PP({\cal H}_i)=\E p_{1i}p_{2i}$ and the bound 
$\E A_1A_2B_i^2({\bar{\bf I}}_{1i}+{\bar{\bf I}}_{2i})=o(1)$.

\begin{proof}[Proof of Theorem \ref{T3}] 
In order to prove  (\ref{NH}) we write $cl(r)=p^*_r/({\overline p}^*_r+p^*_r)$, where
\begin{displaymath}
 p^*(r)=\PP(v_1\sim v_2, d_{12}=r),
\qquad
{\overline p}^*_r=\PP(v_1\not\sim v_2, d_{12}=r), 
\end{displaymath}
and invoke relations 
\begin{eqnarray}\label{p*r}
&&
p^*_r=\E f_{\Lambda}(r)a_1^2B_m^2n^{-1}+o(n^{-1}),
\qquad
r=0,1,\dots,
\\
\label{p*}
&&
{\overline p}^*_0=1-O(n^{-1}),
\qquad
{\overline p}^*_r=O(n^{-2}),
\qquad
r\ge 2.
\\
\label{p**}
&&
{\overline p}^*_1=n^{-1}a_1^2a_2b_2^2+o(n^{-1}).
\end{eqnarray}
Here we denote $\Lambda=a_1B_m\beta_n^{-1/2}$.

Let us prove (\ref{p*}).
For $r=0$ we write
\begin{displaymath}
 \PP(v_1\not\sim v_2, d_{12}=0)=\PP(v_1\not\sim v_2)-\PP(v_1\not\sim v_2,\, d_{12}\ge 1)
\end{displaymath}
and invoke the bounds 
\begin{displaymath}
 1-\PP(v_1\not\sim v_2)=O(n^{-1}),
\qquad
\PP(v_1\not\sim v_2,\, d_{12}\ge 1)=O(n^{-1}).
\end{displaymath}
The first bound follows from (\ref{p-e-inh}).  In order to show the second bound 
we note that the event $\{v_1\not\sim v_2,\, d_{12}\ge 1\}$ implies that there exist 
$i,j\in [m]$, 
$i\not=j$ and $3\le k\le n$ such that 
${\mathbb I}_{1i}{\mathbb I}_{ki}{\mathbb I}_{2j}{\mathbb I}_{kj}=1$. 
Hence, by Markov's inequality,
\begin{displaymath}
\PP(v_1\not\sim v_2,\, d_{12}\ge 1)
\le 
\E\sum_{3\le k\le n}\sum_{i,j\in [m], i\not= j} 
{\mathbb I}_{1i}{\mathbb I}_{ki}{\mathbb I}_{2j}{\mathbb I}_{kj}
=
\sum_{3\le k\le n}\sum_{i,j\in [m], i\not= j} 
\E p_{1i}p_{ki}p_{2j}p_{kj}.
\end{displaymath}
By the inequality 
$\E p_{1i}p_{ki}p_{2j}p_{kj}\le (nm)^{-2}\E A_1A_2A_k^2B_i^2B_j^2$, 
the right hand side  sum is $O(n^{-1})$.

For $r\ge 2$ we write ${\overline p}^*_r\le {\overline p}$, where
${\overline p}=\PP(v_1\not\sim v_2, d_{12}\ge 2)$, and invoke the bound 
${\overline p}=O(n^{-2})$. Let us prove this bound. Given $3\le s<t\le n$  introduce events

${\cal U}_{1,st}= \{\exists i\not=j$ such that 
$w_i\in D_1\cap D_s\cap D_t$ and
$w_j\in D_2\cap D_s\cap D_t\}$;

${\cal U}_{2,st}= \{\exists i_1\not= i_2\not=j$ such that
$w_{i_1}\in D_1\cap D_s$, $w_{i_2}\in D_1\cap D_t$  and
$w_j\in D_2\cap D_s\cap D_t\}$;
 
${\cal U}_{3,st}=\{\exists  i\not=j_1\not=j_2$ such that 
$w_{i}\in D_1\cap D_s\cap D_t$  and
$w_{j_1}\in D_2\cap D_s$, $w_{j_2}\in D_2\cap D_t\}$;
 
${\cal U}_{4,st}=\{\exists  i_1\not=i_2\not=j_1\not=j_2$  such that 
$w_{i_1}\in D_1\cap D_s$, $w_{i_2}\in D_1\cap D_t$  and
$w_{j_1}\in D_2\cap D_s$, $w_{j_2}\in D_2\cap D_t\}$

and observe that ${\cal U}_{st}=\cup_{k\in[4]}{\cal U}_{k,st}$. Next, using the identity
$\{v_1\not\sim v_2, d_{12}\ge 2\}=\cup_{s<t}{\cal U}_{st}$ we obtain
\begin{equation}\label{U-st}
{\overline p}
=
\PP(\cup_{s<t}{\cal U}_{st})
\le
\sum_{s<t}\PP({\cal U}_{st})
=
{\tbinom{n-2}{2}}\sum_{k\in [4]} \PP({\cal U}_{k,st})=O(n^{-2}).
\end{equation}
In the last step  we invoke the bounds that follow
 by Markov's inequality
\begin{eqnarray} \nonumber
&&
\PP({\cal U}_{1,st})\le 
\E
\sum_{i,j\in[m],\, i\not=j}
{\mathbb I}_{1i}{\mathbb I}_{si}{\mathbb I}_{ti}
{\mathbb I}_{2j}{\mathbb I}_{sj}{\mathbb I}_{tj}
\le \frac{a_1^2a_2^2b_3^2}{n^3m},
\\
\nonumber
&&
\PP({\cal U}_{3,st})=\PP({\cal U}_{2,st})\le 
\E
\sum_{i_1,i_2,j\in[m],\, i_1\not=i_2\not=j}
{\mathbb I}_{1i_1}{\mathbb I}_{si_1}{\mathbb I}_{1i_2}{\mathbb I}_{ti_2}
{\mathbb I}_{2j}{\mathbb I}_{sj}{\mathbb I}_{tj}
\le \frac{a_1a_2^3b_2^2b_3}{n^{7/2}m^{1/2}},
\\
\nonumber
&&
\PP({\cal U}_{4,st})\le 
\sum_{i_1,i_2,j_1,j_2\in[m],\, i_1\not= i_2\not=j_1\not=j_2}
{\mathbb I}_{1i_1}{\mathbb I}_{si_1}{\mathbb I}_{1i_2}{\mathbb I}_{ti_2}
{\mathbb I}_{2j_1}{\mathbb I}_{sj_1}{\mathbb I}_{2j_2}{\mathbb I}_{tj_2}
\le \frac{a_2^4b_2^4}{n^4}.
\end{eqnarray}
Proof of (\ref{p*}) is complete.


Let us prove (\ref{p**}).
We have, see (\ref{U-st}),
\begin{displaymath}
 {\overline p}^*_1=\PP(v_1\not\sim v_2, d_{12}\ge 1)-\PP(v_1\not\sim v_2, d_{12}\ge 2)
=
\PP(v_1\not\sim v_2, d_{12}\ge 1)-O(n^{-2}).
\end{displaymath}
Furthermore, from the identity 
$\{v_1\not\sim v_2, d_{12}\ge 1\}=\cup_{3\le s\le n}{\cal U}_s$
we obtain, by inclusion-exclusion,
\begin{displaymath}
 0
\le
\sum_{3\le s\le n}\PP({\cal U}_s)-\PP(v_1\not\sim v_2, d_{12}\ge 1)
\le 
\sum_{3\le s<t\le n}\PP({\cal U}_{st})=O(n^{-2}).
\end{displaymath}
In the last step we used (\ref{U-st}). It remains to evaluate the sum
$\sum_{3\le s\le n}\PP({\cal U}_s)=(n-2)\PP({\cal U}_3)$.
We observe that ${\cal U}_3=\cup_{(i,j)\in M}{\cal H}_{ij3}$, where
$M$ is the set of vectors $(i,j)\in[m]^2$ satisfying $i\not=j$.
We write, by inclusion-exclusion,
$S_1-S_2\le \PP({\cal U}_3)\le S_1$, where
\begin{displaymath}
S_1= \sum_{(i,j)\in M}\PP({\cal H}_{ij3}),
\qquad
S_2=\sum_{(i,j),(k,l)\in M, \,  (i,j)\not=(k,l)}\PP({\cal H}_{ij3}\cap {\cal H}_{kl3}), 
\end{displaymath}
and complete the proof of (\ref{p**}) by showing that
\begin{equation}\label{2013-04-08-4}
 S_1=n^{-2}a_1^2a_2b_2^2(1+o(1)),
\qquad
S_2=o(n^{-2}).
\end{equation}
The first relation of (\ref{2013-04-08-4}) follows from the identity
\begin{displaymath}
\PP({\cal H}_{ij3})=\E p_{1i}p_{2j}p_{3i}p_{3j}=(nm)^{-2}a_1^2a_2b_2^2(1+o(1)),
\end{displaymath}
which is 
obtained using the same truncation argument as in 
 (\ref{sest-06-1}) above. The  second bound of (\ref{2013-04-08-4}) follows from 
the inequalities
 that hold for any $\varepsilon>0$
\begin{eqnarray}\label{2013-04-08-5}
&&
\PP({\cal H}_{ij3}\cap{\cal H}_{kl3})\le c_*\varepsilon^2n^{-2}m^{-4}+o(n^{-3}m^{-3}),
\\
\label{2013-04-08-6}
&&
\PP({\cal H}_{ij3}\cap{\cal H}_{kj3})\le c_*\varepsilon n^{-2}m^{-3}+o((nm)^{-5/2}).
\end{eqnarray}
In order to show (\ref{2013-04-08-5})
we write $1={\mathbb I}+{\overline{\mathbb I}}$, where 
${\mathbb I}={\mathbb I}_{\{A_3\le \varepsilon n\}}$ and invoke the inequalities
\begin{displaymath}
 p_{3k}p_{3l}\le p_{3k}p_{3l}{\mathbb I}+ {\overline{\mathbb I}}
\le 
\frac{\varepsilon^2n^2B_kB_l}{nm}+{\overline {\mathbb I}}
\end{displaymath}
in the identity
$\PP({\cal H}_{ij3}\cap{\cal H}_{kl3})
=
\E p_{1i}p_{2j}p_{3i}p_{3j}p_{1k}p_{2l}p_{3k}p_{3l}  
$. Here we also use the bound $\E A_3^2{\overline{\mathbb I}}=o(1)$.
To show (\ref{2013-04-08-6}) we invoke the inequality 
\begin{displaymath}
p_{3k}
\le 
p_{3k}{\mathbb I}+{\overline{\mathbb I}}
\le 
\frac{\varepsilon n B_k}{\sqrt{nm}}+{\overline{\mathbb I}} 
\end{displaymath}
in the identity 
$\PP({\cal H}_{ij3}\cap{\cal H}_{kj3})=\E p_{1i}p_{2j}p_{3i}p_{3j}p_{1k}p_{3k}$.
Proof of (\ref{p**}) is complete.



Now we prove (\ref{p*r}). Firstly, from  relations 
${\cal H}^*\subset\{v_1\sim v_2\}\subset{\cal H}^*\cup{\cal H}^{**}$
we derive inequalities
\begin{equation}\label{04-22+1}
 0\le \PP(v_1\sim v_2, d_{12}=r)-\PP({\cal H}^*\cap\{d_{12}=r\})\le \PP({\cal H}^{**}).
\end{equation}
Here $\PP({\cal H}^{**})\le \tbinom{m}{2}\PP({\cal H}_{ij})=O(n^{-2})$, since
$\PP({\cal H}_{ij})\le \E p_{1i}p_{1j}p_{2i}p_{2j}\le a_2^2b_2^2(mn)^{-2}$.
Secondly, we write, by symmetry,
\begin{equation}\label{04-22++1}
 \PP({\cal H}^*\cap\{d_{12}=r\})=\sum_{j\in [m]}\PP({\cal H}^*_j\cap\{d_{12}=r\})
=m\PP({\cal H}_m^*\cap \{d_{12}=r\})
\end{equation}
and approximate 
$\PP({\cal H}_m^*\cap \{d_{12}=r\})$ by 
${\tilde p}_r$. We remark that relations
\begin{displaymath}
{\cal H}^*_m\cap\{{\mathbb S}=r\}
\,
\subset 
\,
{\cal H}^*_m\cap\{d_{12}=r\}
\,
\subset
\,
{\cal H}^*_m\cap \left(\{{\mathbb S}=r\}\cup {\cal H}^{***}\right)
\end{displaymath}
imply inequalities
\begin{equation}\label{04-22+2}
 0
\le 
\PP({\cal H}_m^*\cap \{d_{12}=r\})-{\tilde p}_r
\le
\PP({\cal H}_m^*\cap {\cal H}^{***})
\end{equation}
and observe that the probability
\begin{equation}\label{04-22+3}
 \PP({\cal H}_m^*\cap {\cal H}^{***})
\le
(n-2)(m-1)(m-2)\PP({\cal H}^*_{m}\cap{\cal H}_{ijk})
=O(m^{-1}n^{-2})
\end{equation}
 because 
\begin{displaymath}
 \PP({\cal H}^*_{m}\cap{\cal H}_{123})
=
\E
p_{1m}p_{2m}p_{ki}p_{kj}p_{1i}p_{2j}
\le
a_2^3b_2^3(nm)^{-3}.
\end{displaymath}

It follows from (\ref{04-22+1}), (\ref{04-22++1}),  (\ref{04-22+2}), (\ref{04-22+3})
that
\begin{equation}\nonumber
\PP(v_1\sim v_2, d_{12}=r)
=
m{\tilde p}_r
+
O(n^{-2}). 
\end{equation}

We complete the proof of (\ref{p*r}) by showing that 
\begin{equation}\label{04-22+6}
 {\tilde p}_r=\E f_{\Lambda}(r)A_1A_2B_m^2(mn)^{-1}+o(n^{-2}).
\end{equation}
Let us show (\ref{04-22+6}). Using LeCam's inequality, see (\ref{LeCam}), we write
\begin{equation}\label{04-22+5}
 |\PP^*({\mathbb S}=r)-f_{\Lambda_0}(r)|\le \Delta,
\qquad
\Lambda_0:=\sum_{3\le k\le n}p^*_{km},
\qquad
\Delta:=\sum_{3\le k\le n}(p^*_{km})^2.
\end{equation}
Here $p^*_{km}=\E^*{\mathbb I}_{km}=\E^*p_{km}\le a_1B_m(nm)^{-1/2}$. In particular, we have
$\Delta\le a_1^2B_m^2m^{-1}$. This inequality and (\ref{04-22+5}) imply 
\begin{eqnarray}\label{04-22+7}
{\tilde p}_r
&
=
&
\E \PP^*({\mathbb S}=r){\mathbb I}_{{\cal H}^*_m}
=
\E f_{\Lambda_0}(r){\mathbb I}_{{\cal H}^*_m}+R_1,
\\
\nonumber
|R_1|
&
\le
&
\E\Delta{\mathbb I}_{{\cal H}^*_m}
=
\E \Delta p_{1m}p_{2m}
\le 
\E \Delta p_{1m}(A_2n^{-1/2}+{\mathbb I}_{\{B_m>\sqrt{m}\}})=o(n^{-2}).
\end{eqnarray}
Here we used inequalities 
\begin{eqnarray}\nonumber
&&
 p_{2m}\le p_{2m}({\mathbb I}_{\{B_m\le \sqrt{m}\}}+{\mathbb I}_{\{B_m> \sqrt{m}\}})
\le A_2n^{-1/2}+{\mathbb I}_{\{B_m> \sqrt{m}\}},
\\
\nonumber
&&
 \E \Delta p_{1m}A_2n^{-1/2}
\le
 n^{-1}m^{-3/2} a_1^2\E A_1A_2B_m^3=O(n^{-5/2}), 
\\
\nonumber
&&
\E \Delta p_{1m}{\mathbb  I}_{\{B_m>\sqrt{m}\}} 
\le 
n^{-1/2}m^{-3/2}a_1^2\E A_1B_m^3{\mathbb  I}_{\{B_m>\sqrt{m}\}} 
=o(n^{-2}).
\end{eqnarray}
Let us now evaluate the term $\E f_{\Lambda_0}(r){\mathbb I}_{{\cal H}^*_m}$ of 
(\ref{04-22+7}). From relations
\begin{displaymath}
 {\cal H}^*_m\subset{\cal H}_m\subset {\cal H}^*_m\cup{\tilde{\cal H}},
\qquad
{\tilde{\cal H}}:=\cup_{i\in[m-1]}{\cal H}_{im}
\end{displaymath}
we obtain inequalities 
$0\le {\mathbb I}_{{\cal H}_m}-{\mathbb I}_{{\cal H}^*_m}\le {\mathbb I}_{{\tilde {\cal H}}}$
which yield the approximation
\begin{eqnarray}
\label{04-22+8}
&&
\E f_{\Lambda_0}(r){\mathbb I}_{{\cal H}^*_m}=\E f_{\Lambda_0}(r){\mathbb I}_{{\cal H}_m}+R_2,
\\
\nonumber
&&
|R_2|
\le
\sum_{i\in[m-1]}\PP({\cal H}_{im})
=
\sum_{i\in[m-1]}\E p_{1i}p_{2i}p_{1m}p_{2m}
\le 
n^{-2}m^{-1}a_2^2b_2^2.
\end{eqnarray}
Furthermore,  we have
\begin{equation}\label{04-23+1}
 \E f_{\Lambda_0}(r){\mathbb I}_{{\cal H}_m}
=
\E f_{\Lambda_0}(r)p_{1m}p_{2m}
=
\E f_{\Lambda_0}(r)A_1A_2B_m^2(mn)^{-1}+o(n^{-2}).
\end{equation}
In the last step we replaced $p_{1m}p_{2m}$ by $A_1A_2B_m^2(mn)^{-1}$  as in (\ref{sest-06-1}) above.

Now we are going to replace $f_{\Lambda_0}(r)$ by $f_{\Lambda}(r)$. 
For this purpose we combine  the mean value theorem and the inequality 
$|\frac{\partial}{\partial \lambda}f_{\lambda}(r)|\le 1$. We obtain
\begin{equation}\label{04-22+9}
 |f_{\Lambda}(r)-f_{\Lambda_0}(r)|
\le 
|\Lambda-\Lambda_0|.
\end{equation}
Furthermore, we write $\Lambda_0=(n-2)\E^*p_{3m}$ and $\Lambda=n \E ^*(A_3B_m/\sqrt{nm})$
and estimate
\begin{displaymath}
 |\Lambda-\Lambda_0|
\le 
(n-2)|\E^*(\frac{A_3B_m}{\sqrt{nm}}-p_{3m})|+2\frac{a_1B_m}{\sqrt{nm}}
\le (n-2)\E^*\frac{A_3B_m}{\sqrt{nm}}{\bar{\mathbb I}}_{3m}+2\frac{a_1B_m}{\sqrt{nm}}.
\end{displaymath}
The latter inequalities and (\ref{04-22+9}) yield
\begin{equation}\label{04-22+10}
\E f_{\Lambda_0}(r)A_1A_2B_m^2(mn)^{-1}
=
\E f_{\Lambda}(r)A_1A_2B_m^2(mn)^{-1}+o(n^{-2}),
\end{equation}
since 
$\E A_1A_2A_3B_m^3{\bar{\mathbb I}}_{3m}=o(1)$.
Finally, (\ref{04-22+7}), (\ref{04-22+8}), (\ref{04-23+1}) and (\ref{04-22+10}) 
imply (\ref{04-22+6}).
Proof of (\ref{p*r}) is complete.



Let us prove (\ref{clcoef}). To this aim we write $\alpha=\PP({\cal B})/\PP({\cal D})$, where
  ${\cal D}$ denotes the event $\{v_1\sim v_3, v_2\sim v_3\}$ and
${\cal B}={\cal D}\cap\{v_1\sim v_2\}$, 
and show that
\begin{equation}\label{BA17}
 \PP({\cal B})=\varkappa_1+o(n^{-2}),
\qquad
\PP({\cal D})=\varkappa_1+\varkappa_2+o(n^{-2}).
\end{equation}
Here $\varkappa_1:=a_1^3b_3n^{-3/2}m^{-1/2}$ and 
$\varkappa_2:=a_1^2a_2b_2^2n^{-2}$.
To show the first relation of (\ref{BA17}) we observe that  event 
${\cal L}$
implies 
${\cal B}$ and   event ${\cal B}$ implies ${\cal L}\cup{\cal L}^*$.
In particular, we have 
$0\le \PP({\cal B})-\PP({\cal L})\le \PP({\cal L}^*)$. Here
\begin{displaymath}
 \PP({\cal L}^*)
=\tbinom{m}{3} \PP({\cal L}_{123})
\le
\tbinom{m}{3}a_2^3b_2^3(nm)^{-3}
=
O(n^{-3}).
\end{displaymath}
Hence, $\PP({\cal B})=\PP({\cal L})+O(n^{-3})$.
Next we approximate $\PP({\cal L})$ using inclusion-exclusion
\begin{equation}\label{B1-17}
\sum_{s\in [m]}\PP({\cal L}_s)-\sum_{\{s,t\}\subset [m]}\PP({\cal L}_s\cap{\cal L}_t)
\le
\PP({\cal L})
\le
\sum_{s\in [m]}\PP({\cal L}_s)
\end{equation}
and obtain $\PP({\cal L})=\varkappa_1+o(n^{-2})$. Here we invoked the bound  
\begin{displaymath}
\sum_{\{s,t\}\subset [m]}\PP({\cal L}_s\cap{\cal L}_t)
=
\tbinom{m}{2}\PP({\cal L}_1\cap{\cal L}_2)
\le
\tbinom{m}{2}
a_2^3b_3^2(nm)^{-3}
=
O(n^{-4})
\end{displaymath}
and approximated, see (\ref{sest-06-1}),
\begin{displaymath}
\sum_{s\in [m]}\PP({\cal L}_s)=m\PP({\cal L}_s)
=m\E p_{1s}p_{2s}p_{3s}
=
m\bigl(a_1^3b_3(nm)^{-3/2}+o(n^{-3})\bigr)
=
\varkappa_1+o(n^{-2}). 
\end{displaymath}

Let us prove the second relation of (\ref{BA17}). We observe that
${\cal D}={\cal L}\cup{\cal L}^{**}$
and approximate 
\begin{displaymath}
\PP({\cal D})
\approx
\PP({\cal L})+\PP({\cal L}^{**})
\approx 
m\PP({\cal L}_1)+m(m-1)\PP({\cal H}_{123})=\varkappa_1+\varkappa_2+o(n^{-2}).
\end{displaymath}
 Our
rigorous proof is a bit more involved since we operate under minimal moment conditions.
Introduce event ${\cal A}^*=\{A_3<n^{1/4}\}$ and its indicator function 
${\mathbb I}_{{\cal A}^*}$. 
We derive upper and lower bounds for $\PP({\cal D})$ from the inequalities 
\begin{displaymath}
  \PP({\cal L}\cap {\cal A}^*)+\PP( {\cal L}^{**}\cap {\cal A}^*)
-
\PP({\cal L}\cap {\cal L}^{**}\cap {\cal A}^*)
\le 
\PP({\cal D}\cap{\cal A}^*)
\le 
\PP({\cal D})\le \PP({\cal L})+\PP( {\cal L}^{**}).
\end{displaymath}
By the union bound, the right hand side is bounded from above by
\begin{displaymath}
m\PP({\cal L}_1)+m(m-1)\PP({\cal H}_{123})=\varkappa_1+\varkappa_2+o(n^{-2}). 
\end{displaymath}
Next we show a matching lower bound for $\PP({\cal D})$. Proceeding as in (\ref{B1-17}) we write
\begin{displaymath}
 \PP({\cal L}\cap {\cal A}^*)=m\PP({\cal L}_1\cap {\cal A}^*)+O(n^{-4}),
\end{displaymath}
where 
\begin{displaymath}
 \PP({\cal L}_1\cap {\cal A}^*)
=
\E p_{11}p_{21}p_{31}{\mathbb I}_{{\cal A}^*}
=
\E p_{11}p_{21}p_{31}+o(n^{-3})
=
a_1^3b_3(nm)^{-3/2}+o(n^{-3}).
\end{displaymath}
Hence, we have $\PP({\cal L}\cap {\cal A}^*)=\varkappa_1+o(n^{-2})$. It remains to show that
\begin{equation}\label{lija}
 \PP({\cal L}^{**}\cap {\cal A}^*)\ge \varkappa_2+o(n^{-2}),
\qquad
\PP({\cal L}\cap {\cal L}^{**}\cap {\cal A}^*)=o(n^{-2}).
\end{equation}
Let us prove the first inequality of (\ref{lija}). We write, by inclusion-exclusion,
\begin{displaymath}
 \PP( {\cal L}^{**}\cap {\cal A}^*)
\ge 
S_3-S_4,
\quad \
S_3:=\sum_*\PP({\cal H}_{st3}\cap{\cal A}^*),
\quad \
S_4:=\sum_{**}\PP({\cal H}_{st3}\cap{\cal H}_{xy3}\cap{\cal A}^*).
\end{displaymath}
Here and below $\sum_*$ denotes the sum over all  vectors $(s,t)$ with $s\not= t$, 
$1\le s,t\le m$.
By $\sum _{**}$ we denote the sum over unordered pairs of distinct vectors
$\{(s,t), (x,y)\}$ with $s\not=t$, $x\not=y$ and $1\le s,t,x,y\le m$.
Next, we  calculate
\begin{displaymath}
 S_3
=
m(m-1)\PP({\cal H}_{st3}\cap {\cal A}^*)
=
m(m-1)\bigl(a_1^2a_2b_2^2(nm)^{-2}+o((nm)^{-2})\bigr)
=
\varkappa_2+o(n^{-2})
\end{displaymath}
and estimate
\begin{displaymath}
S_4
=
m(m-1)\bigl( (m-2)R_3+\tbinom{m-2}{2}R_4\bigr)
=o(n^{-2}).
\end{displaymath}
Here
\begin{eqnarray}
\nonumber
&&
R_3
=
\PP({\cal H}_{st3}\cap{\cal H}_{sy3}\cap{\cal A}^*)
\le (nm)^{-3}\E A_1A_2^2A_*^3(B_sB_tB_y)^2=O(n^{-2.75}m^{-3}).
\\
\nonumber
&&
 R_4=\PP({\cal H}_{st3}\cap{\cal H}_{xy3}\cap {\cal A}^*)
\le
(nm)^{-4}\E (A_1A_2)^2A_*^4(B_sB_tB_xB_y)^2=O(n^{-3.5}m^{-4}),
\end{eqnarray}
In the last step we used inequalities 
$A_3^4{\mathbb I}_{{\cal A}^*}< A_3^2n^{1/2}$ 
and 
$A_3^3{\mathbb I}_{{\cal A}^*}< A_3^2n^{1/4}$.

It remains to prove the second bound of (\ref{lija}). We apply the union bound 
\begin{displaymath}
 \PP({\cal L}\cap {\cal L}^{**}\cap {\cal A}^*)
\le
\sum_*\PP({\cal L}\cap {\cal H}_{st3}\cap {\cal A}^*)
\le
\sum_*(r_s+r_t+\sum_{u\in [m]\setminus\{s,t\}}r'_{u})
\end{displaymath}
where 
\begin{displaymath}
r_s=\PP({\cal L}_s\cap {\cal H}_{st3}\cap {\cal A}^*), 
\qquad
r_t=\PP({\cal L}_t\cap {\cal H}_{st3}\cap {\cal A}^*),
\qquad
r'_{u}=\PP({\cal L}_u\cap {\cal H}_{st3}\cap {\cal A}^*)
\end{displaymath}
 satisfy $r_s=r_t$ and estimate
\begin{eqnarray}\nonumber
&&
 r_s
\le 
(nm)^{-5/2}\E A_1A_2^2A_3^2B_s^3B_t^2
=
O((nm)^{-5/2}),
\\
\nonumber
&&
r'_{u}
\le 
(nm)^{-7/2}\E A_1^2A_2^2A_3^3B_s^2B_t^2B_u^3{\mathbb I}_{{\cal A}^*}
=O((nm)^{-7/2}n^{1/4}).
\end{eqnarray}
\end{proof}

{\it Acknowledgement}. 
Research 
was supported in part by the  Research Council of Lithuania grant MIP-067/2013.  




\begin{thebibliography}{}

\bibitem{Bakshy}
E. Bakshy, 	I. Rosenn, 	C. Marlow and	L. Adamic,
The role of social networks in information diffusion,
in: Proceedings of the 21st international conference on World Wide Web, WWW 2012, April 16–20, 2012, Lyon, France.
(2012),  519-528. 
ACM 978-1-4503-1229-5/12/04.



\bibitem{Barbour2011}
A. D. Barbour and G. Reinert,
The shortest distance in random multi-type intersection graphs,
Random Structures and Algorithms
39 (2011),
179--209.

\bibitem{Barrat2000}
A. Barrat and M. Weigt,
On the properties of small-world networks,
The European Physical Journal B
13 (2000),
547--560.


\bibitem{behrisch2007}
M.~Behrisch, Component evolution in random intersection graphs, The
Electronical Journal of Combinatorics 14 (2007), $\#$R17.

\bibitem{Blackburn2009}
S. Blackburn and S. Gerke,
Connectivity of the uniform random intersection graph,
Discrete Mathematics,
309 (2009),
5130-5140.


\bibitem{Bloznelis2008}
M. Bloznelis,
Degree distribution of a typical vertex in a general random intersection graph,
Lithuanian Mathematical Journal
 48 (2008), 38--45.



 
\bibitem{Bloznelis2011+}
M. Bloznelis,
Degree and clustering coefficient in sparse random intersection graphs, 
The Annals of Applied Probability 23 (2013), 1254--1289.



\bibitem{BloznelisD2012+}
Bloznelis, M., Damarackas, J. (2012): Degree distribution of 
an inhomogeneous random intersection graph. Submitted, http://arxiv.org/abs/1212.6402 


\bibitem{Bradonjic2010}
M. Bradonjic, A. Hagberg, N. W. Hengartner, A. G. Percus, Component Evolution in General Random Intersection Graphs, The 7th Workshop on Algorithms and Models for the Web Graph, WAW2010. 
Lecture Notes in Computer Science (Springer-Verlag, Berlin, 2010), Vol. 6516, pp. 36-49. 

\bibitem{Britton2008}
T. Britton, M. Deijfen, M. Lindholm, and N. A. Lageras, 
Epidemics on random graphs with tunable clustering.  J. Appl. Prob. 45  (2008), 743--756.

\bibitem{Deijfen}
    { M. Deijfen and W. Kets},
    {Random intersection graphs with tunable degree distribution and clustering},
Probab. Engrg. Inform. Sci. 23 (2009), 661--674.

\bibitem{Durret2007}
R. Durret, {\it Random Graph Dynamics}, Cambridge University Press, 2007.

\bibitem{eschenauer2002}
   L. Eschenauer  and V. D. Gligor,
   A key-management scheme for distributed  sensor networks,
   in: Proceedings of the $9$th ACM Conference on Computer and Communications
   Security
   (2002),
  41--47.


\bibitem{Foudalis2011}
   I. Foudalis, K. Jain, C. Papadimitriou, and M. Sideri,
   Modeling social networks through user background and behavior,
   in: Algorithms and Models for the Web Graph. 
Proceedings of the $8$th International Workshop,WAW 2011, Lecture notes in computer science 6732,
(2011), 85--102.



\bibitem{godehardt2003}
   E. Godehardt and J. Jaworski,
   Two models of random intersection graphs for classification, in:
   Studies in Classification, Data Analysis and Knowledge Organization,
   Springer,
   Berlin--Heidelberg--New York,
   2003,
   67--81.


\bibitem{Guillaume+L2004}
J. L. Guillaume, M. Latapy, Bipartite structure of all complex networks, 
Inform. Process. Lett. 90 (2004) 215--221.




   \bibitem{karonski1999}
   M. Karo\'nski, E. R.  Scheinerman, and K. B. Singer-Cohen,
   On random intersection graphs: The subgraph problem,
   Combinatorics, Probability and Computing
   8 (1999),
   131--159.



\bibitem{Newman2001}
M. E. J. Newman, S. H. Strogatz, and D. J. Watts,
Random graphs with arbitrary degree distributions and their applications,
Physical Review E
64 (2001) 
026118.

\bibitem{Newman2003}
M. E. J. Newman, Properties of highly clustered networks, 
Physical Review E 
68 (2003) 
026121.



\bibitem{Newman+W+S2002}
M. E. J. Newman, D. J. Watts, and S. H. Strogatz,
 Random graph models of social networks,
Proc. Natl. Acad. Sci. USA, 99 (Suppl. 1) (2002), 2566--2572.


\bibitem{Spirakis2011}
S. Nikoletseas, C. Raptopoulos, and P. G. Spirakis,
On the independence number and Hamiltonicity of uniform random
intersection graphs,
Theoretical Computer Science 412 (2011), 6750--6760.




\bibitem{Rybarczyk2011}
 K. Rybarczyk,
Diameter, connectivity, and phase transition of the uniform random intersection graph,
Discrete Mathematics  311 (2011), 1998--2019.






\bibitem{Shang2010}
Y. Shang, Degree distributions in general random intersection graphs, The
Electronical Journal of Combinatorics 17 (2010), $\#$R23.

\bibitem{stark2004}
D. Stark,
The vertex degree distribution of random intersection graphs,
Random Structures and Algorithms
24 (2004),
249--258.

\bibitem{Steele}
J. M. Steele,
Le Cam's inequality and Poisson approximations,
The American Mathematical Monthly 
101 (1994), 48--54. 

\bibitem{storgatz1998}
S. H. Strogatz and D. J. Watts,
Collective dynamics of small-world networks,
Nature,
393 (1998),
440--442.


\bibitem{Traud}
A. L. Traud, E. D. Kelsic, P. J Muchta and M. A. Porter, 
Community structure in online collegiate social networks. 
tech. rep. ArXiv:0809.0690 (September 2008), http://www.amath.unc.edu/Faculty/mucha/reprints/facebook.pdf.


\bibitem{Yagan2009}
O. Yagan and  A. M. Makowski, 
Random key graphs -- can they be small worlds? in:
2009 First International Conference on Networks $\&$ Communications (2009)
313--318.


\bibitem{actornetwork}Information courtesy of the internet movie database, http://www.imdb.com. 


\end{thebibliography}
\end{document}